%% file: AGKR_QPH.tex
\documentclass[a4paper]{article}

\input{setup.tex}
\usepackage[onehalfspacing]{setspace}

\title{Quantum Polynomial Hierarchies: Karp-Lipton, error reduction, and lower bounds}
\date{}
 \author{Avantika Agarwal\footnote{David R. Cheriton School of Computer Science and Institute for Quantum Computing, University of Waterloo, Canada. Email: a243agarwal@uwaterloo.ca. } \and Sevag Gharibian\footnote{Department of Computer Science and Institute for Photonic Quantum Systems (PhoQS), Paderborn University, Germany. Email: \{sevag.gharibian, dorian.rudolph\}@upb.de.} \and Venkata Koppula\footnote{Department of Computer Science and Engineering, Indian Institute of Technology Delhi, India. Email: kvenkata@iitd.ac.in.} \and Dorian Rudolph\footnotemark[2]}

\begin{document}

\maketitle

\begin{abstract}
  The Polynomial-Time Hierarchy (\PH) is a staple of classical complexity theory, with applications spanning randomized computation to circuit lower bounds to ``quantum advantage'' analyses for near-term quantum computers. 
  Quantumly, however, despite the fact that at least \emph{four} definitions of quantum \PH exist, it has been challenging to prove analogues for these of even basic facts from \PH. 
  This work studies three quantum-verifier based generalizations of PH, two of which are from [Gharibian, Santha, Sikora, Sundaram, Yirka, 2022] and use classical strings (\QCPH) and quantum mixed states (\QPH) as proofs, and one of which is new to this work, utilizing quantum pure states (\QPHpure) as proofs.
  We first resolve several open problems from [GSSSY22], including a collapse theorem and a Karp-Lipton theorem for \QCPH. Then, for our new class \QPHpure, we show one-sided error reduction \QPHpure, as well as the first bounds relating these quantum variants of PH, namely $\QCPH\subseteq \QPHpure \subseteq \EXP^{\PP}$. 
\end{abstract}

\section{Introduction}\label{scn:intro}
Introduced by Stockmeyer in 1976~\cite{stockmeyerPolynomialtimeHierarchy1976}, the Polynomial-Time Hierarchy (PH) is one of the foundation stones of classical complexity theory.
Intuitively, the levels of $\PH$, denoted $\Sigma^p_i$ (respectively, $\Pi^p_i$) for $i\geq 1$, yield progressively harder, yet natural, ``steps up'' from \NP (respectively, \coNP). 
Specifically, a $\Sigma^p_i$ verifier is a deterministic poly-time Turing Machine $M$ which, given input $x\in\set{0,1}^n$, takes in $i$ proofs $y_i\in\set{0,1}^{\poly(n)}$, and satisfies:
\begin{align}
  \text{if $x$ is a YES input:}\qquad\exists y_1 \forall y_2 \exists y_3 \cdots Q_iy_i\text{ s.t. }M(x,y_1,\ldots, y_i)=1\\
  \text{if $x$ is a NO input:}\qquad\forall y_1 \exists y_2 \forall y_3 \cdots \overline{Q_i}y_i\text{ s.t. }M(x,y_1,\ldots, y_i)=0.
\end{align}
\noindent Above, $Q_i$ is $\forall$ ($\exists$) if $i$ is even (odd).
\PH has played a prominent (and often surprising!) role in capturing the complexity of various computing setups, including the power of BPP~\cite{sipserComplexityTheoreticApproach1983,lautemannBPPPolynomialHierarchy1983}, low-depth classical circuits~\cite{furstParityCircuitsPolynomialtime1984}, counting classes~\cite{todaPPHardPolynomialTime1991}, and even near-term quantum computers~\cite{bremnerClassicalSimulationCommuting2010,aaronsonComputationalComplexityLinear2011,boulandComplexityVerificationQuantum2019}.

In contrast, the role of \emph{quantum} analogues of \PH in {quantum} complexity theory remains embarrassingly unknown. So, where is the bottleneck? \emph{Defining} ``quantum PH'' is not the problem --- indeed, Yamakami~\cite{yamakamiQuantumNPQuantum2002}, Lockhart and Gonz\'{a}lez-Guill\'{e}n~\cite{j.lockhartQuantumStateIsomorphism2017}, and Gharibian, Santha, Sikora, Sundaram and Yirka~\cite{gharibianQuantumGeneralizationsPolynomial2022} all gave different definitions of quantum PH.
Instead, the difficulty lies in proving even basic properties of quantum PH, which often runs up against difficult phenomena lurking about open problems such as $\exists\boldsymbol{\cdot}\BPP \stackrel{?}{=} \MA$ and $\QMA\stackrel{?}{=}\QMAt$.

In this work, we resolve open questions regarding some fundamental properties of quantum PH. 
We focus on three definitions of quantum PH, chosen because they naturally generalize\footnote{\QCMA and \QMA are quantum generalizations of Merlin-Arthur (\MA), with a classical proof and quantum verifier and a quantum proof and quantum verifier, respectively.} \QCMA and \QMA.
The first two definitions are from~\cite{gharibianQuantumGeneralizationsPolynomial2022} (formal definitions in \Cref{scn:defs}), and the third is new to this work. The definitions all use a poly-time uniformly generated quantum verifier $V$, and are given as follows (for brevity, here we only state the YES case definitions):
\begin{itemize}
    \item \textbf{\QCPH}: $\exists y_1 \forall y_2 \exists y_3 \cdots Q_iy_i$ s.t. $V(x,y_1,\ldots, y_i)$ outputs $1$ with probability $\geq 2/3$.
    \item \textbf{\QPH}: $\exists \rho_1 \forall \rho_2 \exists \rho_3 \cdots Q_i\rho_i$ s.t. $V(x,\rho_1,\ldots, \rho_i)$ outputs $1$ with probability $\geq 2/3$.
    \item \textbf{\QPHpure}: $\exists \ket{\psi_1} \forall \ket{\psi_2} \exists \ket{\psi_3} \cdots Q_i\ket{\psi_i}$ s.t. $V(x,\ket{\psi_1},\ldots, \ket{\psi_i})$ outputs $1$ with probability $\geq 2/3$.
\end{itemize}
In words, \QCPH, \QPH, and \QPHpure utilize poly-size quantum verifiers taking in classical, mixed quantum, and pure quantum proofs, respectively. It is immediate from the definitions that $\QCMA\subseteq\QCPH$, $\QMA\subseteq\QPH$, and $\QMA\subseteq\QPHpure$. Beyond this, not much is clear. For example, a standard use of PH is via its collapse theorem --- if for any $i$, $\Sigma_i^p=\Pi_i^p$, then $\PH=\Sigma_i^p$. Do any of $\QCPH$, $\QPH$, or $\QPHpure$ satisfy such a collapse theorem? Does error reduction hold for $\QPH$ or $\QPHpure$? What is the relationship between $\QCPH$, $\QPH$, and $\QPHpure$? Note that standard convexity arguments (as used for e.g. \QMA) cannot be used to argue $\QPH=\QPHpure$, due to the presence of alternating quantifiers (which make the verification non-convex in the proofs). Can one recover celebrated results for these hierarchies analogous to the Karp-Lipton~\cite{karpConnectionsNonuniformUniform1980} Theorem for PH?

\paragraph{Our results.} We prove the following properties of quantum PH.\\
\vspace{-1mm}

\noindent{\emph{1. Collapse Theorem for \QCPH.}} We first resolve an open question of \cite{gharibianQuantumGeneralizationsPolynomial2022} by giving a collapse theorem for $\QCPH$.

\begin{restatable}{theorem}{theoremQCPHcollapse}\label{thm:collapse}
  If for any $k\geq 1$, $\QCSigma_k = \QCPi_k$, then $\QCPH = \QCSigma_k$.
\end{restatable}
\noindent This is in contrast to \QPH, for which a collapse theorem is believed difficult to show, as it would imply a subsequent collapse\footnote{$\QMAt$ is QMA with two proofs in tensor product. Since its introduction in 2001 by Kobayashi, Matsumoto, and Yamakami~\cite{kobayashiQuantumCertificateVerification2001}, its complexity remains stubbornly open. The current best bounds are $\QMA\subseteq\QMAt\subseteq\QSigma_3\subseteq \NEXP$, where the second and third containments are from\cite{gharibianQuantumGeneralizationsPolynomial2022}.} $\QMAt\subseteq \PSPACE$. Theorem \ref{thm:collapse} immediately resolves a second open question of \cite{gharibianQuantumGeneralizationsPolynomial2022}, recovering a ``precise''\footnote{\PrQCMA is \QCMA but with exponentially small completeness-soundness gap.} quantum version of the Karp-Lipton Theorem:

\begin{restatable}{corollary}{preciseKL}\label{cor:karp-lipton}
  If $\PrQCMA \subseteq \Bpoly$, then $\QCPH = \QCSigma_2$.
\end{restatable}
\noindent Actually, we will do better in this work --- we next leverage \Cref{thm:collapse} and other techniques to obtain a \emph{genuine} (i.e. non-precise) analogue of Karp-Lipton.\\
\vspace{-1mm}

\noindent{\emph{2. Quantum-Classical Karp-Lipton Theorem for \QCPH.}} The celebrated Karp-Lipton theorem~\cite{karpConnectionsNonuniformUniform1980} states that if SAT can be solved by polynomial-size circuits, then PH collapses to $\Sigma_2^p$. Here, we show a quantum analogue for \QCPH:

\begin{restatable}{theorem}{QCKL}\textup{(Karp-Lipton for \QCPH)}\label{thm:qckl}
  If $\QCMA \subseteq \Bpoly$, then $\QCPH=\QCSigma_2=\QCPi_2$.
\end{restatable}
\noindent Here, $\Bpoly$ is $\BQP$ with poly-size classical advice (\Cref{def:Bpoly}). In words, \Cref{thm:qckl} says $\QCMA$ cannot be solved by (even non-uniformly generated) poly-size quantum circuits, unless $\QCPH$ collapses to its second level.\\
\vspace{-1mm}

\noindent{\emph{3. Error reduction for \QPHpure.}} While error reduction for \QCPH follows from parallel repetition (due to its classical proofs), achieving it for \QPHpure is non-trivial for the same reason it is non-trivial for $\QMAt$ --- the tensor product structure between proofs is not necessarily preserved when postselecting on measurements across proof copies in the NO case. Here, we show \emph{one-sided} error reduction for \QPHpure (e.g. \emph{exponentially} small soundness):

\begin{restatable}{theorem}{theoremOneSided}\label{thm:onesided}
  For all $i>0$ and $c-s \geq 1/p(n)$ for some polynomial $p$,
  \begin{enumerate}
      \item For even $i>0$:
        \begin{enumerate}
          \item $\QSigmapure_i(c, s) \subseteq \QSigmapureSEP_i(1/np(n)^2, 1/e^n)$
          \item $\QPipure_i(c, s) \subseteq \QPipureSEP_i(1-1/e^n, 1-1/np(n)^2)$
        \end{enumerate}
      \item For odd $i>0$:
      \begin{enumerate}
        \item $\QSigmapure_i(c, s) \subseteq \QSigmapureSEP_i(1-1/e^n, 1-1/np(n)^2)$
        \item $\QPipure_i(c, s) \subseteq \QPipureSEP_i(1/np(n)^2, 1/e^n)$
      \end{enumerate}
  \end{enumerate}
  \end{restatable}
\noindent Above, $\QSigmapureSEP_i$ and $\QPipureSEP_i$ have the promise that in the YES case, the verifier's acceptance measurement is separable (see \Cref{sscn:errorOneSided}). We remark the proof of this uses a new \emph{asymmetric} version of the Harrow-Montanaro~\cite{harrowTestingProductStates2013a} Product Test, which may be of independent interest (see \Cref{l:apt}). The reason we are unable to recover exponentially small error simultaneously for both completeness and soundness is because our approach requires the final quantifier to be $\exists$.\\
\vspace{-1mm}

\noindent{\emph{4. Upper and lower bounds on $\QPHpure$.}} Having introduced $\QPHpure$ in this work, we next give bounds on its power. 

\begin{theorem}\label{thm:bounds}
  $\QCPH\subseteq \QPHpure \subseteq \EXP^{\PP}$.
\end{theorem}
\noindent While the upper bound above is not difficult to show (\Cref{thm:exptoda}; this may be viewed as an ``exponential analogue'' of Toda's theorem), the lower bound is surprisingly subtle.
The naive strategy of replacing each proof $y_i$ of $\QCPH$ with pure state proof $\ket{\psi_i}$, which is then measured in the standard basis, does not work, as the measurement gives rise to mixed states. Mixed states, in turn, are difficult to handle in $\QPH$, as the latter is not a convex optimization due to alternating quantifiers. 
We remark that while $\QPH\subseteq\QPHpure$ follows easily via purification of proofs, we do \emph{not} know how to show the analogous lower bound $\QCPH\subseteq \QPH$ (our approach uses our asymmetric product test, which requires pure states).

\paragraph{Related work.} Yamakami\cite{yamakamiQuantumNPQuantum2002} gave the first definition of a quantum PH, which takes quantum inputs (in contrast, we use classical inputs). Gharibian and Kempe~\cite{gharibianHardnessApproximationQuantum2012} defined and obtained hardness of approximation results for $\QCSigma_2$, obtaining the first hardness of approximation results for a quantum complexity class. (See \cite{bittelOptimalDepthVariational2023} for a recent extension to QCMA-hardness of approximation results.) Lockhart and Gonz\'{a}lez-Guill\'{e}n\cite{j.lockhartQuantumStateIsomorphism2017} defined a class $\QCPH'$ similar to \QCPH, except using existential and universal \emph{operators}. Thus, in \cite{j.lockhartQuantumStateIsomorphism2017} $\QCSigma'_1=\exists\cdot \BQP$, which is not known to equal \QCMA (for the same reason $\exists\cdot \BPP\stackrel{?}{=}\MA$ remains open). In exchange for not capturing \QCMA, however, the benefit of $\QCPH'$ is that its properties are easier to prove than \QCPH. Gharibian, Santha, Sikora, Sundaram and Yirka~\cite{gharibianQuantumGeneralizationsPolynomial2022} defined $\QCPH$ and $\QPH$, and showed weaker variants of the Karp-Lipton theorem ($\PrQCMA \subseteq \Bpoly$ implies $\QCSigma_2 = \QCPi_2$) and Toda's theorem~\cite{todaPPHardPolynomialTime1991} ($\QCPH\subseteq \p^{\PP^{\PP}}$). They also showed $\QMAt\subseteq \QSigma_3\subseteq \NEXP$, giving the first class sitting between $\QMAt$ and $\NEXP$, and observed that $\QSigma_2=\QPi_2=\textup{QRG}(1)\subseteq\PSPACE$ (due to work of Jain and Watrous~\cite{jainParallelApproximationNoninteractive2009}). Finally, Aaronson, Ingram and Kretschmer~\cite{aaronsonAcrobaticsBQP2022} showed that relative to a random oracle, \PP is not in the ``\QMA hierarchy'', i.e. in $\QMA^{\QMA^{\iddots}}$. The relationship between this \QMA hierarchy and  any of \QCPH, \QPH, or \QPHpure remains open.

The Karp-Lipton theorem has been studied in the setting of quantum advice. Prior works by Aaronson and Drucker~\cite{AaronsonD14}, and Aaronson, Cojocaru, Gheorghiu, and Kashefi~\cite{AaronsonCGK17} studied the consequences of solving NP-complete problems using polynomial-sized quantum circuits with polynomial quantum advice.  

Finally, the Product Test was first introduced by Mintert, Ku\`{s} and Buchleitner~\cite{mintertConcurrenceMixedMultipartite2005}, and rigorously analyzed and strikingly leveraged by Harrow and Montanaro\cite{harrowTestingProductStates2013a} to show $\QMA(k)=\QMA(2)$ for polynomial $k$, as well as error reduction for $\QMAt$. (See also Soleimanifar and Wright~\cite{soleimanifarTestingMatrixProduct2022}.)

\paragraph{Concurrent Work.} We mention two concurrent works on quantum variants of the polynomial hierarchy. First, our collapse theorem for $\QCPH$ (Theorem \ref{thm:collapse} and Corollary \ref{cor:karp-lipton}) was proven concurrently by Falor, Ge, and Natarajan~\cite{FGN23}. Second, we upper bound $\QCPH$ by showing that for all $k$, $\QCPi_k\subseteq\QSigmapure_k\subseteq\QSigmapure_i \subseteq \NEXP^{\NP^{i-1}}$, implying $\QCPH \subseteq \QPHpure \subseteq \EXP^{\PP}$ (Theorem \ref{thm:bounds} and Theorem \ref{thm:exptoda}). Grewal and Yirka~\cite{GY23} show the stronger bound $\QCPH \subseteq \QPH$, at the expense of the minor caveat that their proof does not obtain level-wise containment for all $k$, but rather requires constant factor blowup in level. Beyond this, our papers diverge. We show a Quantum-Classical Karp-Lipton Theorem for QCPH and error reduction for $\QPHpure$. Grewal and Yirka~\cite{GY23} define a new quantum polynomial hierarchy called the \emph{entangled quantum polynomial hierarchy} ($\QEPH$), which allows entanglement across alternatively quantified quantum proofs. They show that $\QEPH$ collapses to its second level (even with polynomially many proofs), and is equal to $\QRGone$, the class of one round quantum-refereed games. They also define a generalization of QCPH, denoted DistributionQCPH, in which proofs are not strings but \emph{distributions} over strings. They show QCPH$=$DistributionQCPH.

\paragraph{Techniques.} We sketch our approach for each result mentioned above.\\

\vspace{-1mm}
\noindent \emph{1. Collapse Theorem for \QCPH.} Collapse theorems for PH are shown via inductive argument --- by fixing an arbitrary proof for the first quantifier of $\Sigma_i^p$, one obtains an instance of  $\Pi_{i-1}^p$. Reference ~\cite{gharibianQuantumGeneralizationsPolynomial2022} noted this approach does not obviously work for $\QCPH$, as fixing the first proof of $\QCSigma_i$ does \emph{not} necessarily yield a valid $\QCPi_{i-1}$ instance (i.e. the latter might not satisfy the desired promise gap). We bypass this obstacle by observing that even if most choices for existentially quantified proofs are problematic, there always exists \emph{at least one} ``good'' choice, for which the recursion works. Formally, we are implicitly using a \emph{promise} version of NP which is robustly\footnote{Formally, let $M$ be a deterministic machine with access to a promise oracle O. We say $M$ is robust~\cite{goldreichPromiseProblemsSurvey2006} if, regardless of how any invalid queries to $O$ (i.e. queries violating the promise gap of $O$) are answered, $M$ returns the same answer. One can define ``PromiseNP'' for non-deterministic $M$ with access to $O$ similarly: In the YES case, $M$ has at least one robust accepting branch, and in the NO case, all branches of $M$ are robust and rejecting. For clarity, we do not formally define and use PromiseNP in this work, but the viewpoint sketched here is equivalent to our approach for \Cref{thm:collapse}.} defined relative to any promise oracle.  \\
\vspace{-1mm}

\noindent{\emph{2. Quantum-Classical Karp-Lipton Theorem.}} The classical Karp-Lipton theorem crucially uses the search-to-decision reduction for SAT. Given a (non-uniform) circuit family that can decide a SAT instance, we can use the circuit family to find a witness for a SAT instance. However, this search-to-decision reduction does not work in the quantum-classical setting since we are working with promise problems instead of languages. As a result, we cannot replicate the classical proof in the quantum-classical setting. Instead, we first convert the $\QCMA$ problem (obtained by fixing the universally quantified proof) to a UniqueQCMA ($\UQCMA$) problem. For this, we use the quantum-classical analogue of Valiant-Vazirani's isolation lemma~\cite{valiantNPEasyDetecting1986} given by Aharonov, Ben-Or, Brand\~{a}o, and Sattath~\cite{ABOBS22}. We then use a single-query (quantum) search-to-decision reduction for $\UQCMA$ that was presented in a recent work by Irani, Natarajan, Nirkhe, Rao and Yuen~\cite{iraniQuantumSearchToDecisionReductions2022}.   

More formally, let $L = (\Lyes, \Lno, \Linv) \in \QCPi_2$, and let $V$ be the corresponding (quantum) verifier. Since the proofs are classical, we assume $V$ has small error. We show that $L \in \QCSigma_2$ by using the following (quantum) verifier $V'$: it takes as input an instance $x$, a quantum circuit $C$, a string $y_1$. The verifier $V'$ runs the circuit $C$ on input $(x,y_1)$ and receives a string $y_2$. It then accepts $x$ if $V$ accepts $(x, y_1, y_2)$. If $x\in \Lno$, then there exists a $y_1$ such that for all $y_2$, $V(x, y_1, y_2)$ accepts with negligible probability. Therefore, for any circuit $C$, there exists a $y_1$ such that $V'(x, C, y_1) = V(x, y_1, C(x, y_1))$ is $1$ with very low probability. 

For any $x \in \Lyes$, we have that for all $y_1$, the pair  $(x, y_1)$ is a YES-instance of a $\QCMA$ problem. Using the \cite{ABOBS22} isolation procedure, we obtain an instance $\phi_{(x, y_1)}$ such that with non-negligible probability $\phi_{(x,y_1)}$ has a unique witness. Next, using the assumption that $\QCMA \subseteq \Bpoly$, we get that there exists a circuit $\tilde{C}$ that can decide $\phi_{(x, y_1)}$. Finally, using the $\UQCMA$ search-to-decision procedure of \cite{iraniQuantumSearchToDecisionReductions2022}, we can use $C$ to find the unique witness for $\phi_{(x, y_1)}$. Let $C$ be the circuit that, on input, $(x, y_1)$, first performs the witness isolation from~\cite{ABOBS22}, followed by the $\UQCMA$ search-to-decision reduction from~\cite{iraniQuantumSearchToDecisionReductions2022}. Putting these together, we get that there exists a circuit $C$ that, for any $y_1$, finds a $y_2$ with non-negligible probability such that $V(x, y_1, y_2) = 1$. Therefore, there exists $C$ such that for any $y_1$, $V'(x, C, y_1) = 1$ with non-negligible probability. \\


\vspace{-1mm}

\noindent{\emph{3. Error reduction for \QPHpure.}} As with $\QMAt$, the challenge with error reduction via parallel repetition for $\QPHpure$ is the following: Given proof $\ket{\psi}_{A_1,B_1}\otimes \ket{\phi}_{A_2,B_2}$, postselecting on a joint measurement outcome on registers $\set{A_1,B_1}$ may entangle registers $\set{A_2,B_2}$. To overcome this, we give an asymmetric version of the Product Test~\cite{harrowTestingProductStates2013a}, denoted APT. The APT takes in an $n$-system state $\ket{\psi}$ in register $A$, and (ideally) $m$ copies of $\ket{\psi}$ in register $B$. It picks a random subsystem $i$ of $A$, as well as a random copy $j$ of $i$ in $B$, and applies the SWAP test (\Cref{fig:SWAP}) between them. We prove (\Cref{l:apt}) that if this test passes with high probability, then $\ket{\psi}_A\approx\ket{\psi_1}\otimes\cdots \otimes\ket{\psi_n}$ for some $\set{\ket{\psi_i}}_{i=1}^n$, i.e. $\ket{\psi}_A$ was of tensor product form. 

With the APT in hand, we can show error reduction for (e.g.) $\QPipure_i$.
Here, the aim of an honest $i$th (existentially quantified) prover is to send many copies of proofs $1$ through $i-1$. With probability $1/2$, the verifier runs the APT with register $A$ being proofs $1$ to $i-1$ and register $B$ being all their copies bundled with the $i$th proof, and with probability $1/2$, the verifier runs parallel repetition on all copies of proofs bundled with the $i$th proof. This crucially leverages the fact that the $i$th proof is existential in the YES case, and thus can be assumed to be of this ideal form. In the NO case, however, the $i$th proof is universally quantified --- thus, we cannot assume anything about its structure, which is why we do not get error reduction for the soundness parameter (in this case).\\
\vspace{-1mm}

\noindent{\emph{4. Upper and lower bounds on $\QPHpure$.}} We focus on the more difficult direction, $\QCPH\subseteq\QPHpure$, for which we actually show that for all even $k\geq 2$, $\QCPi_k\subseteq \QPHpure_k$ (\Cref{l:qcph_in_qphpure}). When simulating \QCPH, the challenge  is again how to deal with {universally} quantified proofs, denoted by index set $U\subseteq[k]$. Unlike existentially quantified proofs, which can be assumed to be set honestly to some optimal string $y_i$, for any index $j\in U$ the proof $\ket{\psi_j}$ can be any pure quantum state. This causes two problems: (1) Measuring $\ket{\psi_j}$ in the standard basis yields a \emph{distribution} $D_j$ over strings, and due to non-convexity of $\QPHpure$, it is not clear if this can help a cheating prover succeed with higher probability that having sent a string. (2) Conditioned on distribution $D_j$, what should the next, existentially quantified, prover $j+1$ set its optimal proof/string to? We overcome these obstacles as follows. The initial setup is similar to \Cref{thm:onesided} --- prover $k$ (which is existentially quantified in the YES case) send copies of all previous proofs $1$ to $k-1$, and with probability $1/2$, we run the APT. However, now with probability $1/2$, we measure \emph{all} proofs in the standard basis. The key step is to immediately \emph{accept} if for any universally quantified proof index $i\in U$, measuring proof $\ket{\psi_i}$ does \emph{not} match all of its copies bundled in proof $k$. In contrast, for existentially quantified proofs, we \emph{reject} if a mismatch occurs. Finally, assuming no mismatches occur, we simply run the original $\QCPH$ verifier on the corresponding strings obtained via measurement. Showing correctness is subtle, and requires a careful analysis for both YES and NO cases, since recall the location of universally quantified proofs changes between cases.

\paragraph{Discussion and open questions.} Many questions remain open for \QCPH, \QPH, and \QPHpure. Perhaps the most frustrating for $\QCPH$ is a lack of a genuine Toda's theorem --- \cite{gharibianQuantumGeneralizationsPolynomial2022} shows $\QCPH\subseteq \p^{\PP^{\PP}}$, but what one really wants is containment in $\p^\PP$. Is this possible? And if not, can one show an oracle separation between $\QCPH$ and $\p^\PP$? Moving to $\QPH$, its role in this mess remains rather murky. Is $\QPHpure \subseteq\QPH$ (recall the converse direction follows via purification)? For this, our proof technique for $\QCPH\subseteq\QPHpure$ appears not to apply, as it requires pure states for the SWAP test. \QPH \emph{does} have one advantage over \QPHpure, however --- while the third level of both contains \QMAt, only $\QSigma_3$ is known to be in \NEXP~\cite{gharibianQuantumGeneralizationsPolynomial2022}, providing a class ``between'' \QMAt and \NEXP. Reference \cite{gharibianQuantumGeneralizationsPolynomial2022}'s proof breaks down for \QPHpure, as its semidefinite-programming approach requires \emph{mixed} state proofs\footnote{Briefly, in $\NEXP$ one can guess the first existentially quantified proof of $\QSigma_3$, leaving a $\QPi_2$ computation. Since we are using mixed states, via duality theory one can rephrase this via an exponential-side SDP, which can be solved in exponential time. Note this ``convexification'' does not seem to apply for larger values of $k$.}. Finally, for \QPHpure, can one show two-sided error reduction? Is there a collapse theorem for \QPHpure? Can one improve our bound $\QPHpure\subseteq \EXP^{\PP}$? As a first step, is $\QSigmapure_2\subseteq\NEXP$? If not, this would suggest the combination of ``unentanglement'' across proofs and alternating quantifiers yields a surprisingly powerful proof system, as it trivially holds that $\QMAt\subseteq \NEXP$.

\paragraph{Organization.} \Cref{scn:defs} begins with notation and definitions. \Cref{sscn:collapse} and \Cref{sscn:qckl} give our collapse theorem and quantum Karp-Lipton theorem for \QCPH, respectively. \Cref{scn:error} shows error reduction for \QPHpure. \Cref{scn:QPHpureupperbound} gives upper and lower bounds for \QPHpure.

\section{Preliminaries}\label{scn:defs}

\paragraph{Notation.}  Let $\operatorname{conv}(S)$ denote the convex hull of set $S$. Then, the set of separable operators acting on $\complex^{d_1}\otimes\cdots\otimes\complex^{d_i}$ is
\begin{align}
  \operatorname{conv}\left(\ketbra{\psi_1}{\psi_1}\otimes\cdots\otimes\ketbra{\psi_i}{\psi_i}\mid \forall j\in[i]\text{, }\ket{\psi_j}\in\complex^{d_j}\text{ is a unit vector} \right).
\end{align}

\begin{definition}[Poly-time uniform family of quantum circuits]
A family of quantum circuits $\{ V_n \}_{n \in \natural}$ is said to be uniformly generated in polynomial time if there exists a polynomially bounded function $t:\natural\mapsto\natural$ and a deterministic Turing machine $M$ acting as follows. For every $n$-bit input $x$, $M$ outputs in time $t(n)$ a description of a quantum circuit $V_n$ (consisting of $1$-and $2$-qubit gates) that takes the all-zeros state as ancilla and outputs a single qubit. We say $V_n$ \emph{accepts} (\emph{rejects}) if measuring its output qubit in the computational basis yields $1$ ($0$).
\end{definition}

\noindent Throughout this paper, we study \emph{promise problems}. A promise problem is a pair $A=(\ayes,\ano)$ such that $\ayes, \ano\subseteq\set{0,1}^\ast$ and $\ayes \cap \ano = \emptyset$, but $\ayes \cup \ano = \set{0, 1}^*$ does not necessarily hold.

\subsection{Quantum-Classical Polynomial Hierarchy (QCPH)} \label{sscn:qcph}
We first recall the quantum analogue of \PH\ that generalizes \QCMA, i.e. has classical proofs~\cite{gharibianQuantumGeneralizationsPolynomial2022}.

\begin{definition}[$\QCSigma_i$]\label{def:QCSigmam}
	Let $A=(\ayes,\ano)$ be a promise problem. We say that $A$ is in $\QCSigma_i(c, s)$ for poly-time computable functions $c, s: \natural \mapsto [0, 1]$ if there exists a poly-bounded function $p:\natural\mapsto\natural$ and a poly-time uniform family of quantum circuits $\{V_n\}_{n \in \natural}$ such that for every $n$-bit input $x$, $V_n$ takes in classical proofs ${y_1}\in \set{0,1}^{p(n)}, \ldots, {y_i}\in \set{0,1}^{p(n)}$ and outputs a single qubit, {such that:}
	\begin{itemize}
		\item Completeness: $x\in \ayes$ $\Rightarrow$ $\exists y_1 \forall y_2 \ldots Q_i y_i$ s.t. $\operatorname{Prob}[V_n \text{ accepts } (y_1, \ldots, y_i)] \geq c$.
		\item Soundness: $x\in \ano$ $\Rightarrow$ $\forall y_1 \exists y_2 \ldots \overline{Q}_i y_i$ s.t. $\operatorname{Prob}[V_n \text{ accepts } (y_1, \ldots, y_i)] \leq s$.
	\end{itemize}
	Here, $Q_i$ equals $\exists$ when $m$ is odd and equals $\forall$ otherwise and $\overline{Q}_i$ is the complementary quantifier to $Q_i$. Finally, define
	\begin{equation}
	\QCSigma_i := \bigcup_{{ c - s \in \Omega(1/\poly(n))}} \QCSigma_i(c, s).
	\end{equation}
\end{definition}

\noindent Comments: Note that the first level of this hierarchy corresponds to $\QCMA$. The complement of the $i^{\mathrm{th}}$ level of the hierarchy, $\QCSigma_i$, is the class $\QCPi_i$ defined next.

\begin{definition}[$\QCPi_i$]\label{def:QCPim}
	Let $A=(\ayes,\ano)$ be a promise problem. We say that $A \in \QCPi_i(c, s)$ for poly-time computable functions $c, s: \natural \mapsto [0, 1]$ if there exists a polynomially bounded function $p:\natural\mapsto\natural$ and a poly-time uniform family of quantum circuits $\{V_n\}_{n \in \natural}$ such that for every $n$-bit input $x$, $V_n$ takes in classical proofs ${y_1}\in \set{0,1}^{p(n)}, \ldots, {y_i}\in \set{0,1}^{p(n)}$ and outputs a single qubit, {such that:}
	\begin{itemize}
		\item Completeness: $x\in \ayes$ $\Rightarrow$ $\forall y_1 \exists y_2 \ldots Q_i y_i$ s.t. $\operatorname{Prob}[V_n \text{ accepts } (y_1, \ldots, y_i)] \geq c$.
		\item Soundness: $x\in \ano$ $\Rightarrow$ $\exists y_1 \forall y_2 \ldots \overline{Q}_i y_i$ s.t. s.t. $\operatorname{Prob}[V_n \text{ accepts } (y_1, \ldots, y_i)] \leq s$.
	\end{itemize}
	Here, $Q_i$ equals $\forall$ when $m$ is odd and equals $\exists$ otherwise, and $\overline{Q}_i$ is the complementary quantifier to $Q_i$. Finally, define
	\begin{equation}
	\QCPi_i := \bigcup_{{c - s \in \Omega(1/\poly(n))}} \QCPi_i(c, s).
	\end{equation}
\end{definition}

\noindent Now the corresponding quantum-classical polynomial hierarchy is defined as follows.

\begin{definition}[Quantum-Classical Polynomial Hierarchy (\QCPH)]\label{def:QCPH}
	\begin{equation}
        \QCPH = \bigcup_{m \in \mathbb{N}} \; \QCSigma_i = \bigcup_{m \in \mathbb{N}} \; \QCPi_i.
    \end{equation}
\end{definition}

\subsection{(Mixed-State) Quantum Polynomial Hierarchy (QPH)}\label{sscn:QPH}

We next define the quantum hierarchy with \emph{mixed state} proofs, as in~\cite{gharibianQuantumGeneralizationsPolynomial2022}. In contrast to \QCPH, this hierarchy generalizes \QMA.

\begin{definition}[$\QSigma_i$]\label{def:QSigmam}
	A promise problem $A=(\ayes,\ano)$ is in $\QSigma_i(c, s)$ for poly-time computable functions $c, s: \natural \mapsto [0, 1]$ if there exists a polynomially bounded function $p:\natural\mapsto\natural$ and a poly-time uniform family of quantum circuits $\{V_n\}_{n \in \natural}$ such that for every $n$-bit input $x$, $V_n$ takes $p(n)$-qubit density operators $\rho_1, \ldots, \rho_i$ as quantum proofs and outputs a single {qubit, then:}
	\begin{itemize}
		\item Completeness: $x\in\ayes$ $\Rightarrow$ $\exists \rho_1 \forall \rho_2 \ldots Q_i \rho_i$ s.t. $\operatorname{Prob}[V_n \text{ accepts } ({\rho_1 \otimes \rho_2 \otimes \cdots \otimes  \rho_i})] \geq c$.
		\item Soundness: If $x\in\ano$ $\Rightarrow$ $\forall \rho_1 \exists \rho_2 \ldots \overline{Q}_i \rho_i$, $\operatorname{Prob}[V_n \text{ accepts } ({\rho_1 \otimes \rho_2 \otimes \cdots \otimes  \rho_i})] \leq s$.
	\end{itemize}
	As before, $Q_i$ equals $\forall$ when $m$ is even and equals $\exists$ otherwise, and $\overline{Q}_i$ is the complementary quantifier to $Q_i$. Define
	\begin{equation}
	   \QSigma_i = \bigcup_{\substack{c-s \in \Omega(1/\poly(n))}} \QSigma_i(c, s).
	\end{equation}
\end{definition}

\noindent \emph{Comments:} (1) {In contrast to the standard quantum circuit model, here we allow \emph{mixed} states as inputs to $V_n$; this can be formally modelled via the mixed state framework of~\cite{aharonovQuantumCircuitsMixed1998}.} (2) Clearly, $\QSigma_1 = \QMA$. (3) We recover the definition of $\QMA(k)$ by ignoring all proofs $\rho_i$ proofs with $i$ even in the definition of $\QSigma_{2k}$.

\begin{definition}[$\QPi_i$]\label{def:QPim}
	A promise problem $A=(\ayes,\ano)$ is in $\QPi_i(c, s)$ for poly-time computable functions $c, s: \natural \mapsto [0, 1]$ if there exists a polynomially bounded function $p:\natural\mapsto\natural$ and a poly-time uniform family of quantum circuits $\{V_n\}_{n \in \natural}$ such that for every $n$-bit input $x$, $V_n$ takes $p(n)$-qubit density operators $\rho_1, \ldots, \rho_i$ as quantum proofs and outputs a single {qubit, then:}
	\begin{itemize}
		\item Completeness: $x\in\ayes$ $\Rightarrow$ $\forall \rho_1 \exists \rho_2 \ldots \overline{Q}_i \rho_i$, $\operatorname{Prob}[V_n \text{ accepts } ({\rho_1 \otimes \rho_2 \otimes \cdots \otimes  \rho_i})] \geq c$.
		\item Soundness: $x\in\ano$ $\Rightarrow$ $\exists \rho_1 \forall \rho_2 \ldots Q_i \rho_i$ s.t. $\operatorname{Prob}[V_n \text{ accepts } ({\rho_1 \otimes \rho_2 \otimes \cdots \otimes  \rho_i})]\leq s$.
	\end{itemize}
	Again, $Q_i$ equals $\exists$ when $m$ is even and equals $\forall$ otherwise, and $\overline{Q}_i$ is the complementary quantifier to $Q_i$. Define
	\begin{equation}
	\QPi_i = \bigcup_{c-s \in \Omega(1/\poly(n))} \QPi_i(c, s).
	\end{equation}
\end{definition}

\begin{definition}[Quantum Polynomial-Hierarchy (\QPH)]\label{def:QPH}
	\begin{equation}
        \QPH = \bigcup_{m \in \mathbb{N}} \; \QSigma_i = \bigcup_{m \in \mathbb{N}} \; \QPi_i.
    \end{equation}
\end{definition}

\subsection{(Pure-State) Quantum Polynomial Hierarchy ($\QPHpure$)}
Above, \QPH\ was defined using mixed-state quantum proofs. This is in contrast to \QMA or \QMAt, where without loss of generality, one may argue due to convexity that pure states suffice. As it is not clear how to extend such convexity arguments to alternating quantifiers, below we introduce Quantum PH with pure-state proofs (for clarity, the definition below is new to this work).

\begin{definition}[$\QSigmapure_i$]\label{def:QSigmam_pure}
	A promise problem $A=(\ayes,\ano)$ is in $\QSigmapure_i(c, s)$ for poly-time computable functions $c, s: \natural \mapsto [0, 1]$ if there exists a polynomially bounded function $p:\natural\mapsto\natural$ and a poly-time uniform family of quantum circuits $\{V_n\}_{n \in \natural}$ such that for every $n$-bit input $x$, $V_n$ takes $p(n)$-qubit states $\ket{\psi_1}, \ldots, \ket{\psi_i}$ as quantum proofs and outputs a single {qubit, then:}
	\begin{itemize}
		\item Completeness: If $x\in\ayes$, then $\exists \ket{\psi_1} \forall \ket{\psi_2} \ldots Q_i \ket{\psi_i}$ such that $V_n$ accepts $({\ket{\psi_1} \otimes \ket{\psi_2} \otimes \cdots \otimes  \ket{\psi_i}})$ with probability $\geq c$.
		\item Soundness: If $x\in\ano$, then $\forall \ket{\psi_1} \exists \ket{\psi_2} \ldots \overline{Q}_i \ket{\psi_i}$ such that $V_n$ accepts $({\ket{\psi_1} \otimes \ket{\psi_2} \otimes \cdots \otimes \ket{\psi_i}})$ with probability $\leq s$.
	\end{itemize}
	Here, $Q_i$ equals $\forall$ when $m$ is even and equals $\exists$ otherwise, and $\overline{Q}_i$ is the complementary quantifier to $Q_i$. Define
	\begin{equation}
    \QSigmapure_i = \bigcup_{\substack{c-s \in \Omega(1/\poly(n))}} \QSigmapure_i(c, s).
	\end{equation}
\end{definition}

\begin{definition}[$\QPipure_i$]\label{def:QPimprime}
	A promise problem $A=(\ayes,\ano)$ is in $\QPipure_i(c, s)$ for poly-time computable functions $c, s: \natural \mapsto [0, 1]$ if there exists a polynomially bounded function $p:\natural\mapsto\natural$ and a poly-time uniform family of quantum circuits $\{V_n\}_{n \in \natural}$ such that for every $n$-bit input $x$, $V_n$ takes $p(n)$-qubit states $\ket{\psi_1}, \ldots, \ket{\psi_i}$ as quantum proofs and outputs a single {qubit, then:}
	\begin{itemize}
		\item Completeness: If $x\in\ayes$, then $\forall \ket{\psi_1} \exists \ket{\psi_2} \ldots Q_i \ket{\psi_i}$ such that $V_n$ accepts $({\ket{\psi_1} \otimes \ket{\psi_2} \otimes \cdots \otimes  \ket{\psi_i}})$ with probability $ \geq c$.
		\item Soundness: If $x\in\ano$, then $\exists \ket{\psi_1} \forall \ket{\psi_2} \ldots \overline{Q}_i \ket{\psi_i}$ such that $V_n$ accepts $({\ket{\psi_1} \otimes \ket{\psi_2} \otimes \cdots \otimes  \ket{\psi_i}})$ with probability $\leq s$.
	\end{itemize}
	Here, $Q_i$ equals $\exists$ when $m$ is even and equals $\forall$ otherwise, and $\overline{Q}_i$ is the complementary quantifier to $Q_i$. Define
	\begin{equation}
	  \QPipure_i = \bigcup_{c-s \in \Omega(1/\poly(n))} \QPipure_i(c, s).
	\end{equation}
\end{definition}

\noindent The Quantum Polynomial Hierarchy with pure-state proofs can now be defined as follows.

\begin{definition}[Pure Quantum Poly-Hierarchy ($\QPHpure$)]\label{def:QPHpure}
	\[
    \QPHpure = \bigcup_{m \in \mathbb{N}} \; \QSigmapure_i = \bigcup_{m \in \mathbb{N}} \; \QPipure_i.
  \]
\end{definition}

\subsection{Other complexity classes}\label{sscn:other}
Next, we define $\PrQCMA$ and $\Bpoly$.

\begin{definition}[$\Bpoly$] \label{def:Bpoly}
	A promise problem $\Pi=(\ayes,\ano)$ is in $\Bpoly$ if there exists a poly-sized family of quantum circuits
	$\{C_n\}_{n \in \mathbb{N}}$ and a collection of binary advice strings $\{a_n\}_{n \in \mathbb{N}}$ with $|a_n| = \poly(n)$, such that for all $n\in\natural$ and all strings $x\in\set{0,1}^n$,
  \[
    \Pr[C_n(\ket{x}, \ket{a_n}) = 1] \geq 2/3 \text{  if  }x \in \ayes\quad\text{  and  }\quad\Pr[C_n(\ket{x}, \ket{a_n}) = 1] \leq 1/3\text{  if  }x \in \ano.
  \]
\end{definition}

\begin{definition}[$\PrQCMA$]\label{def:PrQCMA}
    A promise problem $A=(\ayes,\ano)$ is in $\PrQCMA(c,s)$ for poly-time computable functions $c,s:\natural\mapsto[0,1]$ if there exists poly-bounded functions $p,q:\natural\mapsto\natural$ such that $\forall \ell \in \natural, \ c(\ell) - s(\ell) \geq 2^{-q(\ell)}$, and there exists a poly-time uniform family of quantum circuits $\{V_n\}_{n \in \natural}$ that takes a classical proof ${y}\in \set{0,1}^{p(n)}$ and outputs a single qubit. Moreover, for an $n$-bit input $x$:
    \begin{itemize}
     \item Completeness: If $x\in \ayes$, then $\exists\ y$ such that $V_n$ accepts $y$ with probability at least $c(n)$.
     \item Soundness: If $x\in \ano$, then $\forall\ y$, ${V_x}$ accepts $y$ with probability at most $s(n)$.
    \end{itemize}
    Define $\PrQCMA=\bigcup_{c,s} \PrQCMA(c,s)$.
\end{definition}

\begin{definition}[$\PP$]
	\label{def:PP1}
	A language $L$ is in $\PP$ if there exists a probabilistic polynomial time Turing machine $M$ such that  $x \in L \iff \Pr[M(x) = 1] > 1/2$.
\end{definition}

\section{Collapse theorems and Quantum Karp-Lipton}\label{scn:collapseQKL}
\subsection{Collapse Theorem for QCPH}\label{sscn:collapse}

For the quantum-classical hierarchy, \QCPH, we now show a quantum analogue of the standard collapse theorem for classical PH, i.e. $\Sigmat=\Pit$ implies $\PH=\Sigmat$, resolving an open question of \cite{gharibianQuantumGeneralizationsPolynomial2022}.

\begin{lemma}\label{l:collapse_i}
  If for any $k\geq 1$, $\QCSigma_k = \QCPi_k$, then for all $i\geq  k$, $\QCSigma_i = \QCPi_i=\QCSigma_k$.
\end{lemma}
\begin{proof}
  We proceed by induction. For $j \geq k$, define $P(j) := \QCSigma_j = \QCPi_j = \QCSigma_k$. The base case $P(k)$ holds by the assumption of the lemma. For the inductive case, assume $P(j)$ holds for all $k \leq j \leq i-1$. We show $P(j)$ holds for $j=i$. Consider arbitrary promise problem $L=(\Lyes, \Lno, \Linv) \in \QCSigma_i$ and let $\set{V_n}$ be the verifier circuits for the promise problem. Define new promise problem $L'=(\Lyes', \Lno', \Linv')$:
    \begin{align}
        \Lyes' &= \left\{(x,y_1) \mid \forall y_2 \exists y_3 \ldots Q_i y_i \; \Pr[V(x,y_1, y_2, \ldots, y_i) = 1] \geq \frac{2}{3}\right\}\\
        \Lno' &= \left\{(x,y_1) \mid \exists y_2 \forall y_3 \ldots \overline{Q}_i y_i \; \Pr[V(x,y_1, y_2, \ldots, y_i) = 1] \leq \frac{1}{3}\right\}\\
        \Linv' &= \set{0,1}^* \setminus (\Lyes' \cup \Lno').
    \end{align}
    Clearly, $\Lyes' \cap \Lno' = \emptyset$, and so $(\Lyes', \Lno', \Linv') \in \QCPi_{i-1}$. By the induction hypothesis, there exists promise problem $L''=(\Lyes'', \Lno'', \Linv'')\in \QCSigma_{i-1}$ such that $\Lno' \subseteq \Lno''$ and $\Lyes' \subseteq \Lyes''$. Letting $\set{V''_n}$ denote the verification circuits for $L''$, we have
    \begin{align}
        (x,y_1) \in \Lyes' &&\Rightarrow&& (x,y_1) \in \Lyes'' &&\Rightarrow && \exists y_2 \forall y_3 \ldots Q_i y_i \; \Pr[V''(x,y_1,y_2,\ldots,y_i) = 1] \geq \frac{2}{3}\\
        (x,y_1) \in \Lno' &&\Rightarrow&& (x,y_1) \in \Lno'' &&\Rightarrow &&  \forall y_2 \exists y_3 \ldots \overline{Q}_i y_i \; \Pr[V''(x,y_1,y_2,\ldots,y_i) = 1] \leq \frac{1}{3}.
    \end{align}
    Now considering again $L=(L_{yes}, \Lno, \Linv)$, we have
    \begin{align}
        x \in \Lyes &&\Rightarrow &&\exists y_1\text{ s.t. } (x,y_1) \in \Lyes' &&\Rightarrow &&\exists y_1 \exists y_2 \forall y_3 \ldots Q_{i-1} y_i \Pr[V'(x,y_1,y_2,\ldots,y_i) = 1] \geq \frac{2}{3}\\
        x \in \Lno &&\Rightarrow&& \forall y_1 \text{ s.t. }(x,y_1) \in \Lno' &&\Rightarrow &&\forall y_1 \forall y_2 \exists y_3 \ldots \overline{Q}_{i-1} y_i \Pr[V'(x,y_1,y_2,\ldots,y_i) = 1] \leq \frac{1}{3}.
    \end{align}
    We conclude $L \in \QCSigma_{i-1} = \QCSigma_i$. So $\QCSigma_j = \QCSigma_k$. Similarly, $\QCPi_i \subseteq \QCPi_{i-1} = \QCSigma_k$. Thus, $P(i)$ holds, as claimed.
\end{proof}

\theoremQCPHcollapse*
\begin{proof}
  Immediate from Lemma \ref{l:collapse_i}.
\end{proof}

\noindent Theorem \ref{thm:collapse}, in turn, immediately resolves a second open question of \cite{gharibianQuantumGeneralizationsPolynomial2022}, which showed that if $\PrQCMA \subseteq \Bpoly$, then $\QCSigma_2 = \QCPi_2$: Theorem \ref{thm:collapse} immediately yields that the latter collapses \QCPH.

\preciseKL*
\noindent Recall, however, that later in Section~\ref{sscn:qckl} we will show a stronger version of this result which gives a true quantum-classical analogue of the classical Karp-Lipton Theorem.

\subsection{Quantum-Classical Karp-Lipton Theorem}
\label{sscn:qckl}

We will show that if there exists a polynomial size circuit family $\set{C_n}_{n\in \N}$ that can decide a $\QCMA$ complete problem, then $\QCPi_2 \subseteq \QCSigma_2$. Using Theorem \ref{thm:collapse}, it follows that if $\QCMA \subseteq \Bpoly$, then the quantum-classical polynomial hierarchy collapses to the second level. 

Let $L$ be any promise problem in $\QCPi_2$, and let $\{V_n\}_{n \in \N}$ be the corresponding (quantum) verifier. In order to show $L \in \QCSigma_2$, we need to construct a verifier $V'$ such that if $x \in \Lyes$, then there exists a $y'_1$ such that for all $y'_2$, $V'(x, y_1, y_2)$ accepts with some probability $\alpha$, and if $x\in \Lno$, then for all $y'_1$, there exists a $y'_2$ such that the acceptance probability of $V'(x, y'_1, y'_2)$ is noticeably smaller than $\alpha$. As discussed in the introduction, we require two key ingredients for this result. The first is an isolation-lemma for $\QCMA$, which was given by ~\cite{ABOBS22}. This ensures that a $\QCMA$ instance is converted to a $\UQCMA$ instance. Next, we use the search-to-decision reduction for $\UQCMA$, given by ~\cite{iraniQuantumSearchToDecisionReductions2022}. We first discuss some preliminaries relevant to these two works. 

\paragraph{Quantum-Classical Isolation Lemma and related preliminaries: }

The first step in our proof involves a randomized reduction from $\QCMA$ to $\UQCMA$ ($\QCMA$ but with the additional promise that every `yes' instance has exactly one accepting witness). We recall trivial complete problems for $\QCMA$ and $\UQCMA$, as defined in ~\cite{ABOBS22}, and the randomized reduction from $\QCMA$ to $\UQCMA$.

\begin{definition}[$\UQCMA$]
    Let $A=(\ayes,\ano)$ be a promise problem. We say that $A \in \UQCMA(c,s)$ for polynomial-time computable functions $c, s: \natural \mapsto [0, 1]$ if there exists a polynomially bounded function $p:\natural\mapsto\natural$ and a polynomial-time uniform family of quantum circuits $\{V_n\}_{n \in \natural}$ such that for every $n$-bit input $x$, $V_n$ takes in classical proof ${y_1}\in \set{0,1}^{p(n)}$ and outputs a single qubit, {such that:}
  \begin{itemize}
      \item Completeness: $x\in \ayes$ $\Rightarrow$ $\exists y_1$ s.t. $\Pr[V_n(x, y_1) = 1] \geq c$ and $\forall y_1' \neq y_1 \Pr[V_n(x, y_1) = 1] \leq s$.
      \item Soundness: $x\in \ano$ $\Rightarrow$ $\forall y_1$ s.t. $\Pr[V_n(x, y_1) = 1] \leq s$.
  \end{itemize}
    \begin{equation}
  \text{Define } \; \UQCMA := \bigcup_{{c - s \in \Omega(1/\poly(n))}} \UQCMA(c, s).
  \end{equation}
\end{definition}

\begin{definition}[$\tqcmapp$: Trivial complete problem for $\qcma$]
    \label{def:tqcmapp}
    Let $\tqcmapp$ = $(\tqcmappyes$, $\tqcmappno)$ denote the `trivial' complete problem for $\QCMA$ defined as follows. Every instance consists of a quantum circuit $U$, instance $x$, probabilities $p_1, p_2$, witness length $\ell$ (probabilities and witness length expressed in unary\footnote{We assume the probabilities are rationals, and the numerator and denominator are expressed in unary}). 

    \begin{itemize}
        \item $(U, x, p_1, p_2, \ell) \in \tqcmappyes$ $\implies$ $ \exists d \in \bit^{\ell}$ s.t. 
        $\Pr\Brackets{U\brackets{\ket{x}, \ket{d}} = 1} \geq p_2$

        \item $(U, x, p_1, p_2, \ell) \in \tqcmappno$ $\implies$ $\forall d \in \bit^{\ell}$, 
        $\Pr\Brackets{U\brackets{\ket{x}, \ket{d}} = 1} \leq p_1$
    \end{itemize}

\end{definition}

\begin{definition}[$\uqcmapp$: Trivial complete problem for $\uqcma$]

    Let $\uqcmapp$ = $(\uqcmappyes$, $\uqcmappno)$ denote the `trivial' complete problem for $\UQCMA$ defined as follows. Every instance consists of a quantum circuit $U$, instance $x$, probabilities $p_1, p_2$, witness length $\ell$ (probabilities and witness length expressed in unary). 

    \begin{itemize}
        \item $(U, x, p_1, p_2, \ell) \in \uqcmappyes$ $\implies$ $ \exists d \in \bit^{\ell}$ s.t. 
        $\Pr\Brackets{U\brackets{\ket{x}, \ket{d}} = 1} \geq p_2$
        and for all $d'\neq d$, $\Pr\Brackets{U\brackets{\ket{x}, \ket{d'}} = 1} \leq p_1$

    \item $(U, x, p_1, p_2, \ell) \in \uqcmappno$ $\implies$ $\forall d \in \bit^{\ell}$, 
    $\Pr\Brackets{U\brackets{\ket{x}, \ket{d}} = 1} \leq p_1$
\end{itemize}

\end{definition}

\begin{theorem}[Randomized reduction from $\QCMA$ to $\UQCMA$: Corollary 40 from \cite{ABOBS22}]
\label{uqcma}
There is a probabilistic polynomial time reduction $f$ from $\TQCMAPP$ to $\UQCMAPP$ such that
\begin{itemize}
    \item $\phi = (U, x, p_1, p_2, \ell) \in \tqcmappyes \implies \Pr\Brackets{f(\phi) \in \uqcmappyes} \geq 1/32\ell^2$. 
    \item $\phi = (U, x, p_1, p_2, \ell) \in \tqcmappno \implies \Pr\Brackets{f(\phi) \in \uqcmappno} = 1$
\end{itemize}

Here the probabilities are over the random coins used by reduction $f$. 
\end{theorem}

\paragraph{Search-to-decision Reduction for $\UQCMA$ and related preliminaries: } 
Next, we consider some preliminary results needed for search-to-decision reduction for $\UQCMA$. Since these results were proven either explicitly or implicitly in prior works, we only state the lemma/theorem statements here. 

\begin{lemma}
    \label{measurement}
    Let $\Pi$ be a measurement projector. If $\norm{\ket{\psi} - \ket{\phi}} \leq \epsilon$, then $\trdist{{\Pi \ketbra{\psi}{\psi}} - {\Pi \ketbra{\phi}{\phi}}} \leq 2\epsilon$. 
\end{lemma}

\begin{lemma}[Approximate Bernstein-Vazirani]
    \label{approxbv}
    Consider the function families $f = \{f_n: \bit^n \rightarrow \bit\}$, $g = \{g_n: \bit^n \rightarrow \bit^{t(n)}\}$ and $h = \{h_n: \bit^n \rightarrow \bit\}$. Let the function $f^s_n(x) = \braket{x}{s} \bmod ~ 2$ and $g, h$ be such that $f^s_n(x) = 1 \Leftrightarrow h_{t(n)}(g_n(x)) = 1$. Given a non-uniform quantum circuit family $\{C_n\}$ such that $\{C_{t(n)}\}$ computes $h_{t(n)}$ with probability $1-1/2^n$, there is a non-uniform circuit family $\{C'_n\}$ that takes $\{C_{t(n)}\}$ as advice and outputs $s$ with probability $1-\sqrt{1/2^{n-4}}$.
\end{lemma}

A recent work by Irani \etal~\cite{iraniQuantumSearchToDecisionReductions2022} uses the above approximate Bernstein-Vazirani for a randomized search-to-decision reduction for $\QCMA$.

\QCKL*

\begin{proof}
    Let $L = (\Lyes, \Lno) \in \QCPi_2$. Then there exists a polynomial $\pL:\N \mapsto \N$ and polynomial-time uniform family of quantum circuits $\set{\VLn}_{n\in \N}$ such that  each $\VLn$ takes as input an instance $x \in \bit^n$, two proofs $y_1, y_2 \in \bit^{\pL(n)}$, and the following completeness and soundness guarantees hold:

    \begin{itemize}
        \item (Completeness)

        $x \in \Lyes \Rightarrow \forall y_1$ $\exists y_2$, $\Pr\Brackets{\VLn(x, y_1, y_2) = 1} \geq 1-\frac{1}{\pL(n)^4}$ (where $n=|x|$)

        \item (Soundness)

        $x \in \Lno \Rightarrow \exists y_1$ $\forall y_2$, $\Pr\Brackets{\VLn(x, y_1, y_2) = 1} \leq \frac{1}{\pL(n)^4}$ (where $n=|x|$)
    \end{itemize}

    \vspace{5pt}

    Consider the following promise language $L' = (\Lyes', \Lno')$, which is also characterized by the polynomial $\pL$ and family of quantum circuits $\set{\VLn}_{n\in \N}$ with the following completeness and soundness guarantees:

    \begin{itemize}
        \item (Completeness)

        $(x, y_1) \in \Lyes' \Rightarrow |y_1| = \pL(n) \text{ and }  \exists y_2$ s.t. $\Pr\Brackets{\VLn(x, y_1, y_2) = 1} \geq 1-\frac{1}{\pL(n)^4}$ (where $n=|x|$)

        \item (Soundness)

        $(x, y_1) \in \Lno' \Rightarrow |y_1| = \pL(n) \text{ and }  \forall y_2 \in \bit^{\pL(n)}$  $\Pr\Brackets{\VLn(x, y_1, y_2) = 1} \leq \frac{1}{\pL(n)^4}$ (where $n=|x|$)
    \end{itemize}

    The promise language $L'$ is in $\QCMA$. Therefore, using the (trivial) reduction from $L'$ to $\tqcmapp$, we get the following `yes' and `no' instances for $\tqcmapp$:

    \begin{itemize}
        \item $(x, y_1) \in \Lyes' \Rightarrow \brackets{\VLx, (x, y_1), 1/\pL(|x|)^4, 1-1/\pL(|x|)^4, \pL(|x|)} \in \tqcmappyes $
        \item $(x, y_1) \in \Lno' \Rightarrow \brackets{\VLx, (x, y_1), 1/\pL(|x|)^4, 1-1/\pL(|x|)^4, \pL(|x|)} \in \tqcmappno $
    \end{itemize}

    \vspace{5pt}

    Next, using the randomized reduction from $\QCMA$ to $\UQCMA$ (Theorem \ref{uqcma}), we get a family of verification circuits $\VLxuq$ and probability thresholds $p_1, p_2$ with the following guarantees:

    \begin{itemize}
        \item $(x, y_1) \in \Lyes' \Rightarrow \Pr\Brackets{\brackets{\VLxuq, (x, y_1), p_1, p_2, \pL(|x|)} \in \uqcmappyes}\geq 1/32\pL(|x|)^2$

        \item $(x, y_1) \in \Lno' \Rightarrow \Pr\Brackets{\brackets{\VLxuq, (x, y_1), p_1, p_2, \pL(|x|)} \in \uqcmappno} = 1 $
    \end{itemize}
    Here, the circuit $\VLxuq$ is a `random refinement' of $\VLx$, and $p_1, p_2$ are chosen from an appropriate distribution over $[1/\pL(|x|), 1-1/\pL(|x|)]$.

    \vspace{15pt}

    Suppose $\QCMA \subseteq \Bpoly$, then there exists a family of polynomial-sized circuits $\set{C_n}_{n\in \N}$ and a sequence of (classical) advice strings $\set{\alpha_n}_{n\in \N}$ such that they can decide any yes/no instance of $\TQCMAPP/\UQCMAPP$ with high accuracy. More formally, the following holds:

    \begin{itemize}
        \item If $\phi = (U, x, p_1, p_2, \ell) \in \uqcmappyes$, then $\Pr\Brackets{C_n\brackets{\ket{\alpha_n}, \ket{\phi}} = 1} \geq 1-1/2^{\ell}$
        \item If $\phi = (U, x, p_1, p_2, \ell) \in \uqcmappno$, then $\Pr\Brackets{C_n\brackets{\ket{\alpha_n}, \ket{\phi}} = 1} \leq 1/2^{\ell}$
    \end{itemize}

    Following the single query search-to-decision reduction of Irani \etal ~\cite{iraniQuantumSearchToDecisionReductions2022}, we will use $\set{C_n}_n$ to find a witness for $\UQCMAPP$ instances as follows.

    Let $\phi = (U, x, p_1, p_2, \ell) \in \uqcmappyes$ whose (unique) witness $d_{\phi} \in \bit^{\ell}$ we wish to find. Consider the function $f_{\phi} : \bit^\ell \mapsto \bit$ where $f_{\phi}(z) = \braket{z}{d_{\phi}} (\bmod ~2)$. We will use $\set{C_n}$ to approximate $f_{\phi}(\cdot)$, and then use $\approxbv$ to compute $d_{\phi}$.

    \vspace{5pt}

    Using $\phi$, for every $z\in \bit^\ell$ we can define a new $\uqcmapp$ instance $\phi_z = (U_z, x, p_1, p_2, \ell)$, where $$U_z\brackets{\ket{x}, \ket{d}} = U\brackets{\ket{x}, \ket{d}} \wedge \brackets{\braket{z}{d} (\bmod ~2)} $$

    Since $\phi \in \uqcmappyes$ with unique witness $d_{\phi}$, for any $z\in \bit^\ell$,
    $$f_{\phi}(z) = 1 \iff \phi_z \in \uqcmappyes (\text{with unique witness } d_{\phi})$$
    If $f_{\phi}(z) = 1$, then $\Pr\Brackets{C_{|\phi_z|}\brackets{\phi_z} = 1} \geq 1-1/2^{|z|}$, and if $f_{\phi}(z) = 0$, then $\Pr\Brackets{C_{|\phi_z|}\brackets{\phi_z} = 1} \leq 1/2^{|z|}$. Hence we can use $\approxbv$ (as defined in Lemma \ref{approxbv}) to compute the unique witness if it exists (after running $\approxbv$, we will also check that the output is indeed a witness of $\phi$).

    \vspace{10pt}

    Using the above observations, we can show that $L \in \QCSigma_2$. The verifier $\VLx'$ takes as input the instance $x$, a circuit $C$, string $y_1 \in \bit^{\pL(|x|)}$ and works as follows:
    \begin{enumerate}
        \item Consider the $\tqcmapp$ instance $\brackets{\VLx, (x, y_1), 1/\pL(|x|)^4, 1-1/\pL(|x|)^4, \pL(|x|)}$. The verifier $\VLx'$ first computes the $\uqcmapp$ instance $\phi = \brackets{\VLxuq, (x, y_1), p_1, p_2, \pL(|x|)}$. This instance is in $\uqcmappyes$ with probability at least $1/32\pL(|x|)^2$.
        \item \label{step:approxbv} Next, it uses the circuit $C$ and $\phi$ to run $\approxbv$. In particular, set $f^s_{|x|}$ as $f_{\phi}$ defined previously, $g_{|x|}$ as the trivial reduction from $z$ to $\phi_z$ and $h_{|\phi_z|}(\psi) = 1$ if $\psi \in \uqcmappyes$, $0$ if $\psi \in \uqcmappno$ and undefined otherwise. The circuit $C$ takes as advice the non-uniform circuit family for $\uqcmapp$ and runs $\approxbv$. This results in a string $y_2$.
        \item \label{step:vlx} Finally, it simulates $\VLx(x, y_1, y_2)$ and outputs $1$ iff $\VLx$ does.
    \end{enumerate}

    \vspace{5pt}

    For any $x\in \Lyes$, there exists a circuit $C$ such that for all $y_1 \in \bit^{\pL(|x|)}$,
    \[
        \Pr\Brackets{\VLx'(x, C, y_1) = 1} \geq \frac{1}{32\pL(|x|)^2}\brackets{1-\frac{4}{2^{\pL(|x|)/2}}} \brackets{1-\frac{1}{\pL(|x|)^4}}.
    \]
    On the other hand, if $x\in \Lno$, then there exists a string $y_1$ such that for any string $y_2$ computed in Step \ref{step:approxbv}, $\VLx'$ outputs $1$ with probability at most $1/\pL(|x|)^4$.
\end{proof}

\subsubsection{Immediate Implications of Theorem \ref{thm:qckl}}
We now discuss some immediate implications of the quantum-classical variant of Karp-Lipton Theorem. 

\paragraph{$\UQCMA \subseteq \Bpoly \implies \QCPH$ collapse}

Notice in the proof of Theorem \ref{thm:qckl} that the search-to-decision reduction only uses the non-uniform circuit for deciding $\UQCMAPP$ as advice. In particular this means that we only need the assumption $\UQCMA \subseteq \Bpoly$ for the above proof to work. This leads to the following corollary:

\begin{corollary}
    If $\UQCMA \subseteq \Bpoly$, then $\QCPi_2 \subseteq \QCSigma_2$, and therefore the quantum-classical polynomial hierarchy collapses to the second level.
\end{corollary}

\paragraph{A stronger collapse}
In the classical setting, Sengupta observed that if $\NP \subseteq \P/\poly$, then $\PH$ collapses to $\sptwo$ (this observation is first mentioned in the work of Cai~\cite{C01}). The complexity class $\sptwo$~\cite{RS98,Can96} is contained in $\Sigma_2 \cap \Pi_2$, and therefore gives a `stronger' Karp-Lipton theorem in the classical setting. Following our proof above, we get a similar result for the quantum-classical setting --- $\QCPH$ collapses to $\qcsptwo$ if $\QCMA \subseteq \Bpoly$. Below, we define the complexity class $\qcsptwo$, and state the result. 

\begin{definition}[Quantum-Classical $\sptwo$]
    \label{qcsptwo}
    Let $A=(\ayes,\ano)$ be a promise problem. We say that $A$ is in $\qcsptwo(c, s)$ for polynomial-time computable functions $c, s: \natural \mapsto [0, 1]$ if there exists a polynomially bounded function $p:\natural\mapsto\natural$ and a polynomial-time uniform family of quantum circuits $\{V_n\}_{n \in \natural}$ such that for every $n$-bit instance $x$, $V_n$ takes in classical proofs ${y_1}\in \set{0,1}^{p(n)}, {y_2}\in \set{0,1}^{p(n)}$ and outputs a single qubit, {such that:}
    \begin{itemize}
        \item Completeness: $x\in \ayes$ $\Rightarrow$ $\exists y_1 \forall y_2$, $\Pr[V_n(x, y_1, y_2) = 1] \geq c$.
        \item Soundness: $x\in \ano$ $\Rightarrow$ $\exists y_2 \forall y_1$, $\Pr[V_n(x, y_1, y_2) = 1] \leq s$.
    \end{itemize}

    \begin{equation}
    \text{Define } \; \qcsptwo := \bigcup_{{ c - s \in \Omega(1/\poly(n))}} \qcsptwo(c, s).
    \end{equation}
    
\end{definition}

\begin{theorem}\label{thm:qcklsptwo}
    If $\UQCMA \subseteq \Bpoly$, then $\QCPH \subseteq \qcsptwo$.
\end{theorem}

In the classical setting, the Karp-Lipton collapse was improved to $\ZPP^{\NP}$ ~\cite{BCGKT96,KobWat98}. Later, Cai showed that $\sptwo \subseteq \ZPP^{\NP}$, and the collapse was further improved to `oblivious' $\sptwo$ ($\osptwo$)\cite{CR06}. 

In the quantum setting, we don't have an analogue for Cai's result, and the following are some open questions:
\begin{itemize}
    \item Does $\QCMA \subseteq \Bpoly$ imply that $\QCPH \subseteq \ZPP^{\QCMA}$? 
    \item Is $\qcsptwo$ contained in $\ZPP^{\QCMA}$?
\end{itemize}

\paragraph{Implications for $\QCMA$ and $\QCAM$}

    Arvind \etal~\cite{AKSS95} showed that if $\NP$ languages can be decided using polynomial sized circuits, then $\AM = \MA$. Using our techniques, one can obtain a similar implication for $\QCMA$ and $\QCAM$. This proof does not follow directly by combining the classical result with our $\QCPH$ collapse result. However, our approach for the quantum-classical Karp-Lipton theorem can be modified to get this result. Since this proof is very similar, we state the result here.

    First, we define the complexity class $\QCAM$. 

    \begin{definition}[$\QCAM$]
        Let $A=(\ayes,\ano)$ be a promise problem. We say that $A \in \QCAM(c,s)$ for polynomial-time computable functions $c, s: \natural \mapsto [0, 1]$ if there exists a polynomially bounded function $p:\natural\mapsto\natural$ and a polynomial-time uniform family of quantum circuits $\{V_n\}_{n \in \natural}$ such that for every $n$-bit input $x$, $V_n$ takes in classical proof ${y}\in \set{0,1}^{p(n)}$ and a random string $z \in \set{0,1}^{p(n)}$ and outputs a single qubit, {such that:}
        \begin{itemize}
          \item Completeness: $x\in \ayes$ $\Rightarrow$ $\Pr_z[\exists y_1 \text{ s.t. } \Pr[V_n(x, z, y_1) = 1] \geq 1-\frac{1}{2^n}] \geq c$.
          \item Soundness: $x\in \ano$ $\Rightarrow$ $\Pr_z[\exists y_1 \text{ s.t. } \Pr[V_n(x, z, y_1) = 1] \geq 1-\frac{1}{2^n}] \leq s$.
        \end{itemize}
        \begin{equation}
            \text{Define } \; \QCAM := \bigcup_{{c' - s' \in \Omega(1/\poly(n))}} \QCAM(c, s).
        \end{equation}
    \end{definition}

    \begin{theorem}[Quantum-Classical analogue of \cite{AKSS95}]
        \label{thm:qcma=qcam}
        If $\uqcma \subseteq \Bpoly$, then $\QCMA = \QCAM$.
    \end{theorem}

\section{Error reduction for \QPHpure}\label{scn:error}

We next study (weak) error reduction for \QPHpure (pure proofs). For this, we first require an asymmetric generalization of the Product Test \cite{mintertConcurrenceMixedMultipartite2005,harrowTestingProductStates2013a}, given in Section~\ref{sscn:apt}. We then give one-sided error reduction results in Section~\ref{sscn:errorOneSided}. 

\subsection{Asymmetric product test}\label{sscn:apt}

\begin{figure}[t]
  \[ \Qcircuit @C=1.5em @R=0.5em {
       \lstick{\ket{0}}    &\gate{H} & \ctrl{1}                       & \gate{H} &\meter\\
       \lstick{\ket{\psi}} &\qw      & \multigate{1}{\textup{SWAP}} & \qw      &\qw\\
       \lstick{\ket{\phi}} &\qw      & \ghost{\textup{SWAP}}     &\qw       &\qw\\
  }\]
  \caption{The SWAP test, whose output is the measurement result on the first wire.}
  \label{fig:SWAP}
  \end{figure}
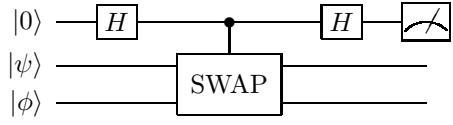

We first give a generalization of the Product Test \cite{mintertConcurrenceMixedMultipartite2005,harrowTestingProductStates2013a}, which we denote the Asymmetric Product Test (APT), stated as follows:
\begin{enumerate}
    \item The input is $\ket{\psi}\in\C^{d_1}\otimes\cdots\otimes\C^{d_n}$ in register $A$, and $\ket{\phi}\in \complex^{d^m}$ in register $B$, where for brevity $d:=d_1\cdots d_n$. We think of $B$ as encoding $m$ copies of $A$.
    \item Choose $(i,j)\in[n]\times [m]$ uniformly at random.
    \item Run the SWAP Test (\Cref{fig:SWAP}) between the $i$th register of $A$, and in $B$, the $i$th register of the $j$th copy of $A$.
    \item Accept if the SWAP Test outputs $0$, reject otherwise.
\end{enumerate}
In fact, above, one can also assume that $\ket{\phi}$ in the APT is potentially entangled across two registers $B$ and $C$, as we do for the main lemma of this section, given below.

\begin{lemma}[Asymmetric Product Test (APT)]\label{l:apt}
  Define $d=d_1\cdots d_n$.
  Consider $\ket{\phi}_{BC}\in\C^{d^m}\otimes\C^{d'}$ for some $d'>0$.
  Suppose
  \begin{equation}
    \max_{\ket{\psi}:= \ket{\psi_1}\otimes\cdots\otimes\ket{\psi_n}\in\C^d}\quad\bra\phi_{BC}[(\ketbra{\psi}{\psi}^{\otimes m})_B\otimes I_C]\ket\phi_{BC}=1-\epsilon
  \end{equation}
  for $\epsilon\geq 0$.
  Then, given the state $\ket{\eta}_{ABC}:=\ket{\psi}_A\otimes \ket{\phi}_{BC}$, the APT accepts with probability at most $1-\epsilon/2mn$.
\end{lemma}
\begin{proof}
  Let $\{\ket{\alpha_{ik}}\}_{k\in [d_i]}$ be an orthonormal basis of $\C^{d_i}$ with $\ket{\alpha_{i1}}:=\ket{\psi_i}$, and define index set
  \begin{equation}
   X = \{(x_{ij})_{i\in[n],j\in[m]} \mid \forall ij:x_{ij} \in [d_i] \}.
  \end{equation}
  Then, rewrite $\ket\phi_{BC}$ in the $\set{\ket{\alpha_{ik}}}_k$ bases to obtain:
  \begin{equation}
    \ket\phi = \sum_{x\in X} a_x \left(\bigotimes_{ij}\ket{\alpha_{i,x_{ij}}}\right)_B\ket{\gamma_x}_C
  \end{equation}
  for some states $\ket{\gamma_x}_C$.
  Without loss of generality, consider the swap test between the first register of $A$, $\complex^{d_1}$ (which contains $\ket{\psi_1}$), and the first copy of $\complex^{d_1}$ in $B$.
  Just prior to the final measurement in the test, we have the state $\ket{\eta'}_{SABC}$ given by (where $S$ encodes the control qubit for SWAP in the SWAP test)
  \begin{align}
    & \sum_{x\in X} a_x\left( \frac12 \ket{0}_S\bigl(\ket{\psi_1}\ket{\alpha_{1,x_{11}}} + \ket{\alpha_{1,x_{11}}}\ket{\psi_1}\bigr) + \frac12 \ket{1}_S\bigl(\ket{\psi_1}\ket{\alpha_{1,x_{11}}} - \ket{\alpha_{1,x_{11}}}\ket{\psi_1}\bigr) \right)\otimes \bigotimes_{i\ne 1 \text{ or }j\neq 1} \ket{\alpha_{i,x_{ij}}}\ket{\gamma_x}.\nonumber\\
    &=: \sum_{x\in X} a_x\ket{\eta'_x}_{SABC}.
  \end{align}
  We now show that the cross terms of $\braket{\eta'}{\eta'}$ vanish, as $\braket{\eta'_x}{\eta'_y}=0$ with $x\ne y$, where $x,y\in X$:
  (1) If $x_{ij} \ne y_{ij}$ for $(i,j)\ne (1,1)$, this follows immediately from orthonormality of the basis sets $\{\ket{\alpha_{ik}}\}_{k\in [d_i]}$.
  (2) The only remaining case is $x_{11}\ne y_{11}$.
  Without loss of generality, $y_{11}\ne 1$.
  Then $\ket{\alpha_{1,y_{11}}}$ is orthogonal to $\ket{\psi_1}$, again by choice of our basis set.
  Hence, $\braket{\psi_1,\alpha_{1,x_{11}}}{\alpha_{1,y_{11}},\psi_1}=0$.
  Since trivially $\braket{\psi_1,\alpha_{1,x_{11}}}{\psi_1,\alpha_{1,y_{11}}}=0$, we again have $\braket{\eta'_x}{\eta'_y}=0$.

  As for the non-cross-terms, we first have again from our basis choice that for any $x\in X$,
  \begin{equation}
    \bra{\eta'_x}(\ketbra{0}{0}_S)\ket{\eta'_x} = \begin{cases}
      1, & \text{if }x_{11}=1\\
      \frac12, & \text{if }x_{11}\ne 1
    \end{cases}.
  \end{equation}
  Recall now that the APT selects $i\in[n]$ and $j\in[m]$ uniformly at random and does a SWAP test between $\ket{\psi_{i}}$ (the $i$th register of $A$) and the $j$th copy of the $i$ register of $B$.
  Thus, conditioned on the APT randomly choosing $(i,j)=(1,1)$, we may bound its acceptance probability as
  \begin{equation}
    \Pr[\APT(\ket\eta)=1\mid i=1,j=1] = \bra{\eta'}(\ketbra{0}{0}_S)\ket{\eta'} = \sum_{\substack{x\in X\\\text{s.t. }x_{11}=1}} \abs{a_x}^2 + \frac12\sum_{\substack{x\in X\\\text{s.t. }x_{11}\neq1}} \abs{a_x}^2.
  \end{equation}
  By symmetry, an identical argument holds for any pair $(i,j)$, and so
  \begin{align}
    \Pr[\APT(\ket\eta)=1] &= \frac1{mn}\sum_{ij}\left(\sum_{\substack{x\in X\\\text{s.t. }x_{ij}=1}} \abs{a_x}^2 + \frac12\sum_{\substack{x\in X\\\text{s.t. }x_{ij}\ne 1}} \abs{a_x}^2\right) \\
    &\le \abs{a_{1^{mn}}}^2 + \frac{2mn-1}{2mn}\sum_{x\in X\setminus\{1^{mn}\}}\abs{a_x}^2\\
    &= 1-\frac{\epsilon}{2mn}.
  \end{align}
\end{proof}

\subsection{One-sided error reduction}\label{sscn:errorOneSided}

With \Cref{l:apt} (APT) in hand, we now show one-sided error reduction for $\QPipure_i$, which suffices to obtain statements for all desired classes subsequently in \Cref{thm:onesided}.
For this, we define classes $\QSigmapureSEP_i$ and $\QPipureSEP_i$ as identical to $\QSigmapure$ and $\QPipure$, respectively, except the measurement POVM of the verifier in the YES case must additionally be a separable operator relative to the cuts between each of the $i$ proofs $\ket{\psi_1}$ to $\ket{\psi_i}$.
(This is analogous to $\QMAt$ versus $\QMAtSEP$ \cite{harrowTestingProductStates2013a}.)

\begin{lemma}[One-sided $\QPipure_i$ amplification]\label{l:qpi}
If $i$ is even, then
\begin{equation}
  \QPipure_i(c, s) \subseteq \QPipureSEP_i\left(1-\frac{1}{e^n}, 1-\frac{1}{np(n)^2}\right),
\end{equation}
for all functions $c$ and $s$ such that $c - s \geq$ 1/p(n) for some polynomial $p$.
\end{lemma}
\begin{proof}
Let L = $(L_{yes}, \Lno, \Linv) \in \QPipure_i(c, s)$, with verifier $V$ taking in $i$ proofs denoted $\ket{\psi_1},\ldots \ket{\psi_i}$. Since $i$ is even, the last proof, $\ket{\psi_i}$, is existentially quantified. We define a new verifier $V'$ to decide $L$ in $\QPipure_i(1-1/\exp, 1-1/\poly)$ as follows. $V'$ receives the following proofs from an honest prover:
\begin{align}
    \left(\ket{\psi_1'}\otimes\cdots\otimes\ket{\psi_{i-1}'}\right)_A &= \left(\ket{\psi_1}\otimes\cdots\otimes\ket{\psi_{i-1}}\right)_A\\
    \ket{\psi_i'}_{BC} &= \left(\bigotimes_{j = 1}^{i-1} \ket{\psi_j}\right)_B^{\otimes m}\otimes \ket{\psi_i}_C^{\otimes m}\label{eqn:rep}
\end{align}
for $m \in \Theta(n(c-s)^{-2})$, and where $A$, $B$ and $C$ are used to align with the notation of \Cref{l:apt} (which we use shortly). In words, the last prover sends $m$ copies of the first $i-1$ proofs in register $B$, and $m$ copies of the last proof $\ket{\psi_i}$ in register $C$.
Then, $V'$ acts as follows:
\begin{enumerate}
    \item With probability 1/2, apply the APT (\Cref{l:apt}) between registers $A$ and $B$.
    Accept iff the test accepts.
    \item With probability 1/2, apply verifier V $m$ times, taking one proof $\ket{\psi_i}$ from each respective subregister of $B$.
    Accept iff at least $(c+s)/2$ measurements accept.
\end{enumerate}

\paragraph{Correctness.} We now argue correctness.\\\vspace{-2mm}

\noindent\emph{YES case.} Since the $i^{\mathrm{th}}$ prover is existentially quantified, it sends the state in \Cref{eqn:rep}.
    Thus, the APT accepts with certainty.
    Similarly, for parallel repetition of V, each repetition is independent, hence the overall verifier accepts with probability at least $1 - \exp(-(c-s)^2m/2)$, as desired.

    \noindent\emph{NO case.} Now the $i^{\mathrm{th}}$ prover is universally quantified, hence it can send us a state entangled across $BC$. \\
    Suppose the APT accepts with probability $1 - \epsilon$.
    By \Cref{l:apt},
    \begin{align}
        \ket{\psi_i'}_{BC} = \alpha\left(\bigotimes_{j = 1}^{i-1}\ket{\psi_j}\right)_B^{\otimes m}\ket{\phi}_C + \beta\ket{\gamma}_{BC}=:\alpha\ket{\eta}_B\ket{\phi}_C + \beta\ket{\gamma}_{BC}\quad\text{ for }\quad\abs{\alpha}^2\geq 1-2 m(i-1)\epsilon,
    \end{align}
    for arbitrary states $\ket{\phi}$ and $\ket{\gamma}$ satisfying that $\ket{\eta}_A\ket{\phi}_B$ is orthogonal to $\ket{\gamma}_C$.
    We have
    \begin{align}
        Pr[\text{parallel repetition of $V$ accepts } \ket{\eta}_A\ket{\phi}_B] \leq e^{\frac{-(c-s)^2m}{2}}.
    \end{align}
    We conclude the acceptance probability of $V'$ is at most
    \begin{align}
        \frac{1}{2}\left( (1-\epsilon) + (1-2 m(i-1)\epsilon)e^{-(c-s)^2m/2} + 2m(i-1)\epsilon + 2\sqrt{2 m(i-1)\epsilon + (2 m(i-1))^2\epsilon} \right)
        \le 1 - \frac\epsilon 2
    \end{align}
    where the maximum is attained when $\epsilon = \Theta(1/m(i-1))$.\\\vspace{-1mm}

Finally, that the measurement operator for the YES case is separable follows via the argument of Harrow and Montanaro for $\QMAt$ amplification~\cite{harrowTestingProductStates2013a}, since our use of the APT is agnostic to whether proofs are universally or existentially quantified.
\end{proof}

\Cref{l:qpi} now easily generalizes to cover all classies regarding $\QPHpure$ we are concerned with:
\theoremOneSided*
\begin{proof}
Statement $1b$ is from Lemma \ref{l:qpi}, and $1a$ follows from $1b$ since we can get a $\QSigmapure_i$ verifier by flipping the answer of a $\QPipure_i$ verifier corresponding to the complement of our promise problem. The remaining cases are analogous: $2a$ follows from $1b$, and $2b$ from $1a$.
\end{proof}

\section{Upper and lower bounds on $\QPHpure$}\label{scn:QPHpureupperbound}

\subsection{Lower bound: $\QCPH$ versus $\QPH$}\label{sscn:qph_vs_qcph}

We first give a lower bound on $\QPHpure$, by showing that alternatingly-quantified classical proofs can be replaced by pure-state quantum proofs.

\begin{theorem}\label{thm:qcph_in_qphpure}
  $\QCPH \subseteq \QPHpure$.
\end{theorem}
This follows immediately from the following lemma.
\begin{lemma}\label{l:qcph_in_qphpure}
  For all even $k\geq 2$, $\QCPi_k\subseteq \QSigmapure_k$.
\end{lemma}
\begin{proof}
  That $k$ is even implies the $k$th proof is existentially quantified in the YES case, a fact we will leverage.
  To begin, Let $V$ be a verifier for $\QCPi_k$, so that in the YES case $\forall_1x_1\dots \exists_kx_k:\Pr[V(x_1,\dots,x_k)=1]\ge c$ and in the NO case $\exists x_1\dots \forall x_k:\Pr[V(x_1,\dots,x_k)=1]\le s$, where we may assume without loss of generality that $c$ and $s$ are exponentially close to $1$ and $0$, respectively.
  We construct a $\QPHpure$ verifier $V'$ as follows:
  \begin{itemize}
    \item $V'$ receives $k$ proofs, $\ket{\psi_1}\otimes\cdots\otimes\ket{\psi_k}$.
    \item The last proof $\ket{\psi_k}$ consists of two registers denoted $A$ and $B$. We think of $A$ as containing $m$ copies of proofs $1,\dots,k-1$ (as in the APT, \Cref{l:apt}), and $B$ as containing the $k$th proof for $V$, $x_k$.
    \item $V'$ acts as follows:
      \begin{enumerate}
          \item With probability $1/2$, run the APT (\Cref{l:apt}), and accept if and only if the APT accepts.
          \item With probability $1/2$, measure proofs $\ket{\psi_1}\otimes\cdots\otimes\ket{\psi_{k-1}}$ in the standard basis to obtain strings $x_1,\ldots, x_{k-1}$, respectively, similarly measure all copies of these proofs in $A$, and finally measure $B$ in the standard basis to obtain $x_{k,B}$. Let $U=\set{1,3\ldots, k-1}$ denote the indices of universally quantified proofs (in the YES case).  Then:
          \begin{enumerate}
              \item If there exists an $i\in [k-1]$ such that the strings obtained by measuring all copies of $\ket{\psi_i}$ did \emph{not} equal $x_i$, let $i$ denote the minimal such index. Accept if $i\in U$, and reject otherwise.
              \item Otherwise, simulate $V(x_1,\dots,x_{k-1},x_{k,B})$.
          \end{enumerate}
      \end{enumerate}
\end{itemize}

\paragraph{Correctness strategy.} Since we are trying to simulate $\QCPi_k$, ideally we want all proofs to be strings.
This can be assumed without loss of generality for existentially quantified proofs, but not for universally quantified proofs, which can be set to any pure state by definition of $\QPipure_k$.
So, for any $i\in U$, write $\ket{\psi_i}=\sum_{x}\alpha_{i,x}\ket{x}$, where we view $\abs{\alpha_{i,x}}^2$ as a distribution over strings $x$ for proof $i$.
Denote for any $i\in U$ by $x^*_i$ the amplitude of highest weight, i.e.  $x^*_i=\argmax_{x}{\abs{\alpha_{i,x}}}$ (ties broken arbitrarily).
The key idea is that the existentially quantified proof at index $i+1$ will now send string $y_{i+1}^*$, where $y_{i+1}^*$ is the same string that a prover for the \emph{original} $\QCPi_k$ verifier $V$ would have sent in response to $x^*_i$ on proof $i$. (For clarity, if $i+1=k$, then $y_{i+1}^*$ is sent in register $B$.)\\

\noindent\emph{YES case.} Since the $k$th proof is existentially quantified, for Step 1 (APT), we may assume all copies in $\ket{\psi_{k,A}}$ (the state in register $A$) are correctly set, so the APT accepts with probability $1$, leaving all proofs invariant.
As for Step 2, for each $i\in U$, let $X_i$ be the random variable resulting from measuring $\ket{\psi_i}$ in the standard basis, and likewise $X_{ij}$ the random variables for the $j$th copy of $\ket{\psi_i}$ in register $A$.
The verification can now fail in one of two ways:
\begin{enumerate}
    \item There exists an $i\in U$ such that we did not measure the ``right'' result, i.e. $X_i\neq x_i^*$, but that all measured copies of $\ket{\psi_i}$ returned the same string, i.e. $\forall j$ $X_{ij}=X_i$.
    In this case, our strategy for setting the existentially quantified proof $i+1$ is not necessarily the correct response to $X_i$.
    Thus, when Step $2(b)$ is run, we have no guarantee for the acceptance probability of $V$.
    \item Measuring all $i\in U$ yields the desired outcomes $x_i^*$ (as well as $\forall j$ $X_{ij}=X_i$), but $V$ nevertheless rejects due to imperfect completeness, i.e. $c<1$.
\end{enumerate}

\noindent Combining these, via the union bound we thus have for Step 2 that
  \begin{align}
    \Pr[\text{reject}] &\le \Pr[\;\exists i\in U:X_i\ne x^*_i\text{ AND } \forall ij:X_i=X_{ij}\;] + (1-c) \\
    &\le \min\{p,1-p\}^m + 1-c\\
    &\le 2^{-m}+1-c,
  \end{align}
  where $p:=\max_i\abs{\alpha_{i,x_i}}^2$, since $\Pr[X_i\ne x^*_i]=\Pr[X_{ij}\ne x^*_i]\le \max\set{p,1-p}$ because $\abs{\alpha_{i,y}}^2\le 1-p$ for $y\ne x^*_i$.
  Since we assumed the APT accepts with perfect probability, we conclude $V'$ accepts with probability $\ge\frac12 + \frac12(c - 2^{-m})=:c'$. (As an aside, recall $c$ is exponentially close to $1$.)\\

\noindent\emph{NO case.} The analysis is more subtle in this case, as the set of indices $U=\set{1,3,\ldots, k-1}$ now refers to \emph{existentially} quantified proofs.
Thus, in Step 2(a) when $V'$ accepts iff $i\in U$, this now means it accepts on existentially quantified proofs.
This is because $V'$ does not know whether it is in a YES or NO case.
For the same reason, the actions and role of the final proof $\ket{\psi_k}$ on registers $A$ and $B$ remain the same, even though it is now universally quantified.
Finally, the strategy of any existential prover $i\in U$ is the same as the YES case: Prover $i$ sends the optimal response $y_i^*$ to universally quantified proof $x_{i-1}^*$.
(If $i=1$, then there is no universal proof to condition on for $\ket{\psi_1}$.)

To proceed, assume for now that the APT would have succeeded in Step 1 with certainty.
In Step 2, define again for all $i\in U$, $X_i$ the random variable resulting from measuring $\ket{\psi_i}$ in the standard basis, and $X_{ij}$ the random variables for the $j$th copy of $\ket{\psi_i}$ in register $A$.
The verifier can now fail in one of two ways, the first of which differs significantly from the YES case:
\begin{enumerate}
    \item There exists $i\in U$ such that $X_i$ mismatched one of its copies in $A$, i.e. $\exists j$ such that $X_i\neq X_{ij}$. \emph{A priori}, this seems like a problem --- since $\ket{\psi_k}$ is universally quantified, most choices of $\ket{\psi_k}$ will cause a mismatch with $X_i$ with high probability, causing $V'$ to accept with high probability.
    The crucial insight is that, in order for $\ket{\psi_1}\otimes\cdots\otimes\ket{\psi_k}$ to pass the APT, it must essentially set each copy of $\ket{\psi_i}$ to string $y_i^*$. Thus, measuring the $A$ register is highly unlikely to produce mismatches on existentially quantified proofs!
    \item Measuring all $i\in U$ will yield the desired outcomes $y_i^*$, since $U$ is existentially quantified. If in addition $\forall j$ $X_{ij}=X_i$, running $V$ may nevertheless accept due to imperfect soundness, i.e. $s>0$.
\end{enumerate}
Then, by a similar argument as for the YES case that in Step 2, in which we first assume the APT passes with certainty,
  $
    \Pr[\text{accept}] \le  2^{-m}+s.
  $
  Now, let us factor in the probability of the APT of Step 1 passing.
  Assume the APT accepts with probability $\ge1-\epsilon/2mn$.
  By \Cref{l:apt},
  \begin{equation}
    \bra{\psi_k}(\ketbrab{\psi_1,\dots,\psi_{k-1}}_A\otimes I_B)\ket{\psi_k} \ge 1-\epsilon.
  \end{equation}
  Let $\{\ket{\beta_i}\}$ be an orthonormal basis of register $A$ with $\ket{\beta_1} = \ket{\psi_1,\dots,\psi_{n-1}}$.
  Then we can write $\ket{\psi_n} = \sum_i \alpha_i \ket{\beta_i}_A\ket{\gamma_i}_B$ with $\abs{\alpha_1}^2\ge 1-\epsilon$.
  Hence, $V'$ accepts with probability at most
  \begin{equation}
    \frac12\left(1-\frac{\epsilon}{2mn}\right) + \frac12\left(\epsilon + (1-\epsilon)s\right)\le 1 - \frac{1}{4mn}+s=:s',
  \end{equation}
  where the first inequality follows because, without loss of generality, we may assume $\epsilon \leq 1/2$, as otherwise the prover cannot hope to succeed make $V'$ accept with probability greater than $3/4$ (whereas $c'\approx 1$). Finally, we can choose $c,s,m$ such that $c'-s'\ge1/\poly$.
\end{proof}

\subsection{Upper bound}\label{sscn:upperbound}
To complement \Cref{sscn:qph_vs_qcph}, we next give a simple but non-trivial upper bound on \QPHpure, which may be viewed as an ``exponential analogue'' of Toda's theorem. For this, let $\NP^k$ denote a tower of $\NP$ oracles of height $k$. (For example, $\NP^1=\NP$ and $\NP^2=\NP^{\NP}$.)
Define $\NEXP^k$ analogously, by for a tower of $\NEXP$ oracles.
In \cite{gharibianQuantumGeneralizationsPolynomial2022}, it was observed that $\QSigma_i \subseteq \NEXP^i$. Here, we show a sharper bound.

\begin{theorem}\label{thm:exptoda}
  $\QPHpure \subseteq \EXP^{\PP}$.
\end{theorem}

We prove this using a sequence of statements.

\begin{observation}\label{obs:trivialQPH}
  $\QSigmapure_i \subseteq \NEXP^{\NP^{i-1}}$.
\end{observation}
\begin{proof}
  Replace all proofs by by their exponential-size classical description (up to additive additive inverse exponential additive error in the entries), and simulate the verifier's action on the proofs via exponential-time matrix multiplication. The standard proof technique for showing $\Sigma_i^p\subseteq \NP^i$ now applies, except we only require $\NEXP$ at the base level of the oracle tower, i.e. $\NEXP^{\NP^{i-1}}$, since an exponential time base can ``inflate'' or pad the instance size for its oracle exponentially.
\end{proof}
\noindent For comparison, the onbserved bound in \cite{gharibianQuantumGeneralizationsPolynomial2022} of $\QSigma_i \subseteq \NEXP^i$ is overkill, since it allows the first NEXP oracle can use \emph{double} exponential time to process its exponetial size input.

\begin{observation}\label{obs:nexptoexp}
  $\NEXP \subseteq \EXP^{\NP}$.
\end{observation}
\begin{proof}
Since using an exponential time machine we can ``inflate'' the instance size to exponential, an $\NP$ machine can thereafter simulate the $\NEXP$ computation on the inflated instance size. The $\EXP$ machine just returns the answer of the $\NP$ oracle.
\end{proof}

\begin{observation}\label{obs:NPO}
    $\NEXP^{O} \subseteq \EXP^{\NP^{O}}$ for an oracle to any language $O$.
\end{observation}
\begin{proof}
    It is easy to see that the argument in Observation \ref{obs:nexptoexp} relativizes, since the $\NP$ oracle to the $\EXP$ machine can make the $\NEXP$ queries directly to the oracle $O$.
\end{proof}

\begin{observation}\label{obs:exppp}
    $\EXP^{\p^{\PP}} \subseteq \EXP^{\PP}$.
\end{observation}
\begin{proof}
    The $\EXP$ machine can make at most exponentially many queries to the its oracle, each of which can be of size at most exponential in the size of the input. Therefore an $\EXP$ machine can simulate the action of a $\p^{\PP}$ machine (even on an exponential sized query) by making queries to a $\PP$ oracle while simulating the action of a $\p$ machine (which will only take time polynomial in size of the query).
\end{proof}

\begin{theorem}\label{thm:EXP_PP}
    For all $i\geq 1$, $\QSigma_i \subseteq \EXP^{\PP}$.
\end{theorem}
\begin{proof}
    By \Cref{obs:trivialQPH}, \Cref{obs:NPO} and Toda's Theorem \cite{todaPPHardPolynomialTime1991},
    \begin{align}
        \QSigma_i\subseteq\NEXP^{\NP^{i-1}}\subseteq\EXP^{\NP^i} \subseteq \EXP^{\p^{\PP}}.
    \end{align}
    The claim now follows from Observation \ref{obs:exppp}.
\end{proof}
\noindent \Cref{thm:exptoda} now follows immediately from \Cref{thm:EXP_PP}.


\section*{Acknowledgements}
We thank Chirag Falor, Shu Ge, Anand Natarajan, Sabee Grewal, and Justin Yirka for the pleasure of productive discussions during the concurrent development of our works. This work was completed in part while AA was a student at Indian Institute of Technology Delhi and in part while visiting Paderborn University. SG was supported by
the DFG under grant numbers 450041824 and 432788384, the BMBF within the funding program “Quantum Technologies - from Basic Research to Market” via project PhoQuant (grant
number 13N16103), and the project “PhoQC” from the programme “Profilbildung 2020”, an initiative of the Ministry of Culture and Science of the State of Northrhine Westphalia. VK was supported by the Pankaj Gupta Fellowship at IIT Delhi. 

\printbibliography

\end{document}

%% file: setup.tex
\usepackage{amsthm,amsmath,amssymb,amsfonts}
\usepackage[american]{babel}
\usepackage{graphicx}
\usepackage{hyperref}
\usepackage{cleveref}
\usepackage[url=false,giveninits=true,maxbibnames=99,style=alphabetic,maxalphanames=5]{biblatex}
\usepackage{mathtools}
\usepackage{enumitem}
\usepackage{csquotes}
\usepackage[noend]{algpseudocode}
\usepackage{algorithm}
\usepackage{xparse}
\usepackage{xspace}
\usepackage{color}
\usepackage{caption}
\usepackage{subcaption}
\usepackage{fullpage}
\usepackage{placeins}
\usepackage{xcolor}
\usepackage{makecell}
\usepackage{qcircuit}
\usepackage{thmtools}
\usepackage{thm-restate}
\usepackage[textsize=scriptsize]{todonotes}
\usepackage{verbatim}
\usepackage{qcircuit}
\usepackage{comment}
\usepackage{mathdots}
\newcommand{\snote}[1]{\textcolor{blue}{ {\textbf{(Sev: }#1\textbf{) }}}}

\setuptodonotes{color=blue!15}

\mathchardef\mhyphen="2D 

\newcommand\newmathabbrev[2]{\newcommand{#1}{\ensuremath{#2}\xspace}}

\newcommand\cfont\mathsf
\newmathabbrev\p{\cfont{P}}
\newmathabbrev{\N}{\mathbb N}
\newmathabbrev\NP{\cfont{NP}}
\newmathabbrev\QPH{\cfont{QPH}}
\newmathabbrev\QPHpure{\cfont{pureQPH}}
\newmathabbrev\QCPH{\cfont{QCPH}}
\newmathabbrev\QEPH{\cfont{QEPH}}

\newcommand{\QRGone}{\mathsf{QRG(1)}}

\newmathabbrev\DTIME{\cfont{DTIME}}
\newmathabbrev\tSAT{3\cfont{\mhyphen{}SAT}}
\newmathabbrev\MA{\cfont{MA}}
\newmathabbrev\AM{\cfont{AM}}
\newmathabbrev\NPDAG{\cfont{NP\mhyphen{}DAG}}
\newmathabbrev\QMADAG{\cfont{QMA\mhyphen{}DAG}}
\newmathabbrev\yes{\mathrm{yes}}
\newmathabbrev\no{\mathrm{no}}
\newmathabbrev\US{\cfont{US}}
\newmathabbrev\FP{\cfont{FP}}
\newmathabbrev\PP{\cfont{PP}}
\newmathabbrev\CeP{\cfont{C_=P}}
\newmathabbrev\coCeP{\cfont{coC_=P}}
\newmathabbrev\PH{\cfont{PH}}
\newmathabbrev\SAT{\cfont{SAT}}
\newmathabbrev\SPP{\cfont{SPP}}
\newmathabbrev\GapP{\cfont{GapP}}
\newmathabbrev\BQP{\cfont{BQP}}
\newmathabbrev\QP{\cfont{QP}}
\newmathabbrev\StoqMA{\cfont{StoqMA}}
\newmathabbrev\coNP{\cfont{coNP}}
\newmathabbrev\AzPP{\cfont{A_0PP}}
\newmathabbrev\QMA{\cfont{QMA}}
\newmathabbrev\coQMA{\cfont{coQMA}}
\newmathabbrev\BPP{\cfont{BPP}}
\newmathabbrev\QCMA{\cfont{QCMA}}
\newmathabbrev\pNPlog{\p^{\NP[\log]}}
\newmathabbrev\pNP{\p^{\NP}}
\newmathabbrev\pNPtwo{\p^{\NP[2]}}
\newmathabbrev\pNPone{\p^{\NP[1]}}
\newmathabbrev\pParSAT{\p^{||\SAT}}
\newmathabbrev\pQMApar{\p^{||\QMA}}
\newmathabbrev\pCpar{\p^{||\C}}
\newmathabbrev\pStoqMApar{\p^{||\StoqMA}}
\newmathabbrev\pQMAlog{\p^{\QMA[\log]}}
\newmathabbrev\pClog{\p^{\textup{C}[\log]}}
\newmathabbrev\pC{\p^{\textup{C}}}
\newmathabbrev\QMASPACE{\cfont{QMASPACE}}
\newmathabbrev\pQMAtlog{\p^{\QMA(2)[\log]}}
\newmathabbrev\pStoqMAlog{\p^{\StoqMA[\log]}}
\newmathabbrev\pQMApt{\p^{\Vert\QMA(2)}}
\newmathabbrev\pQMA{\p^{\QMA}}
\newmathabbrev\SharpP{\cfont{\#P}}
\newmathabbrev\pSharP{\p^{\SharpP[1]}}
\newmathabbrev\PromisePP{\cfont{PromisePP}}
\newmathabbrev\lett{\le_\mathrm{tt}}
\newmathabbrev\YES{\mathsf{YES}}
\newmathabbrev\NO{\mathsf{NO}}
\newmathabbrev\PSPACE{\cfont{PSPACE}}
\newmathabbrev\IP{\cfont{IP}}
\newmathabbrev\POLY{\cfont{POLY}}
\newmathabbrev\DAG{\cfont{DAG}}
\newmathabbrev\StoqMADAG{\StoqMA\mhyphen\cfont{DAG}}
\newmathabbrev\CDAG{C\mhyphen\cfont{DAG}}
\newmathabbrev\CDAGf{C\mhyphen\cfont{DAG}_f}
\newmathabbrev\CDAGs{C\mhyphen\cfont{DAG}_s}
\newmathabbrev\CDAGd{C\mhyphen\cfont{DAG}_{d}}
\newmathabbrev\CDAGo{C\mhyphen\cfont{DAG}_1}
\newmathabbrev\LOGS{\cfont{LOGS}}
\newmathabbrev\TAUT{\cfont{TAUTOLOGY}}
\newmathabbrev\SBQP{\cfont{SBQP}}
\newmathabbrev\Fc{F_\coNP}
\newmathabbrev\Fa{F_\AzPP}
\newmathabbrev\GSCON{\cfont{GSCON}}
\newmathabbrev\GSCONexp{\GSCON_\cfont{exp}}
\newmathabbrev\QMAexp{\QMA_\cfont{exp}}
\newmathabbrev\UQMA{\cfont{UQMA}}
\newmathabbrev\R{\mathbb R}
\newmathabbrev\Trees{\cfont{TREES}}
\newmathabbrev\apxsim{\cfont{APX\mhyphen{}SIM}}
\newmathabbrev\AWPP{\cfont{AWPP}}
\newmathabbrev\X{\mathcal{X}}
\newmathabbrev\Y{\mathcal{Y}}

\newmathabbrev\Z{\mathcal{Z}}
\newcommand{\ZPP}{\mathsf{ZPP}}
\newmathabbrev\ZZ{\mathbb{Z}}
\newmathabbrev\Hprop{H_\mathrm{prop}}
\newmathabbrev\Hin{H_\mathrm{in}}
\newmathabbrev\Hout{H_\mathrm{out}}
\newmathabbrev\Hstab{H_\mathrm{stab}}
\newmathabbrev\Lext{\L_\mathrm{ext}}
\newmathabbrev\BTWNP{\cfont{BTW}(\NP)}
\newmathabbrev\BSN{\cfont{BSN}}
\newmathabbrev\SN{\cfont{SN}}
\newmathabbrev\BD{\cfont{BD}}
\newmathabbrev\HYPERTREE{\cfont{NP\mhyphen{}HYPERTREE}}
\newmathabbrev\Hext{H_\mathrm{ext}}
\newmathabbrev\Hpropt{\tilde{H}_\mathrm{prop}}
\newmathabbrev\Hint{\tilde{H}_\mathrm{in}}
\newmathabbrev\Houtt{\tilde H_\mathrm{out}}
\newmathabbrev\EXP{\cfont{EXP}}
\newmathabbrev\A{\mathcal{A}}
\newmathabbrev\U{\mathcal{U}}

\renewcommand\L{\mathcal{L}}

\newmathabbrev\DAGSSAT{\DAGS(\SAT)}
\newmathabbrev\DAGS{\mathrm{DAGS}}
\newmathabbrev\DAGSNP{\DAGS(\NP)}

\newmathabbrev\AND{\cfont{AND}}

\newmathabbrev\STCONN{{S,T}\cfont{\mhyphen{}CONN}}
\newmathabbrev\CNF{\cfont{CNF}}
\newmathabbrev\NEXP{\cfont{NEXP}}
\newmathabbrev\NPSPACE{\cfont{NPSPACE}}
\newmathabbrev\QCMASPACE{\cfont{QCMASPACE}}
\newmathabbrev\BQPSPACE{\cfont{BQPSPACE}}
\newmathabbrev{\PCP}{\cfont{PCP}}
\newmathabbrev\BQUPSPACE{\cfont{BQ_UPSPACE}}
\newmathabbrev\QMAt{\QMA(2)}
\newmathabbrev\QMAtSEP{\QMA^{\mathsf{SEP}}(2)}
\newmathabbrev\QMAtexp{\QMAt_{\exp}}
\newmathabbrev\MIP{\cfont{MIP}}
\newmathabbrev\MIPt{\MIP(2)}
\newmathabbrev\BellQMA{\cfont{BellQMA}}
\newmathabbrev\BellQMAt{\BellQMA(2)}
\newmathabbrev\BellQMAtexp{\BellQMAt_{\exp}}

\protected\def\verythinspace{%
  \ifmmode
    \mskip0.5\thinmuskip
  \else
    \ifhmode
      \kern0.08334em
    \fi
  \fi
}

\newcommand{\C}{\mathbb C}

\newcommand{\be}{\begin{equation}}
\newcommand{\ee}{\end{equation}}

\newcommand{\APT}{\mathrm{APT}}

\renewcommand{\epsilon}{\varepsilon}
\DeclareMathOperator*{\argmax}{arg\,max}

\newcommand{\set}[1]{{\left\{#1\right\}}}    

\DeclareMathOperator{\poly}{poly}

\DeclarePairedDelimiter\bra{\langle}{\rvert}
\DeclarePairedDelimiter\ket{\lvert}{\rangle}

\DeclarePairedDelimiter\abs{\lvert}{\rvert}
\DeclarePairedDelimiter\norm{\lVert}{\rVert}

\DeclarePairedDelimiterX\braket[2]{\langle}{\rangle}{#1 \delimsize\vert #2}
\DeclarePairedDelimiterX\ketbra[2]{\lvert}{\rvert}{#1 \delimsize\rangle\delimsize\langle #2}

\newcommand{\ketbrab}[1]{\ketbra{#1}{#1}}

\setlist[itemize]{noitemsep, topsep=0pt}
\setlist[enumerate]{noitemsep, topsep=0pt}

\declaretheorem[numberwithin=section]{theorem}
\declaretheorem[sibling=theorem]{observation}
\declaretheorem[sibling=theorem]{corollary}
\declaretheorem[sibling=theorem]{lemma}

\declaretheorem[sibling=theorem,style=definition]{definition}

\crefname{observation}{observation}{observations}
\Crefname{observation}{Observation}{Observations}

\newcommand{\ayes}{A_{\textup{yes}}} 
\newcommand{\ano}{A_{\textup{no}}} 

\makeatletter
\newcommand{\subalign}[1]{%
  \vcenter{%
    \Let@ \restore@math@cr \default@tag
    \baselineskip\fontdimen10 \scriptfont\tw@
    \advance\baselineskip\fontdimen12 \scriptfont\tw@
    \lineskip\thr@@\fontdimen8 \scriptfont\thr@@
    \lineskiplimit\lineskip
    \ialign{\hfil$\m@th\scriptstyle##$&$\m@th\scriptstyle{}##$\hfil\crcr
      #1\crcr
    }%
  }%
}
\NewDocumentCommand{\LeftComment}{s m}{%
  \Statex \IfBooleanF{#1}{\hspace*{\ALG@thistlm}}\(\triangleright\) #2}

\def\moverlay{\mathpalette\mov@rlay}
\def\mov@rlay#1#2{\leavevmode\vtop{%
   \baselineskip\z@skip \lineskiplimit-\maxdimen
   \ialign{\hfil$\m@th#1##$\hfil\cr#2\crcr}}}
\newcommand{\charfusion}[3][\mathord]{
    #1{\ifx#1\mathop\vphantom{#2}\fi
        \mathpalette\mov@rlay{#2\cr#3}
      }
    \ifx#1\mathop\expandafter\displaylimits\fi}

\makeatother

\algnewcommand{\LineComment}[1]{\State \(\triangleright\) #1}

\algblockdefx[ON]{Blk}{EndBlk}[1]
  {#1}
  {}

\makeatletter
\ifthenelse{\equal{\ALG@noend}{t}}%
  {\algtext*{EndBlk}}
  {}%
\makeatother

\AtEveryBibitem{%
  \clearlist{language}%
}

\newcommand{\Lyes}{{L}_{\mathrm{yes}}}
\newcommand{\Lno}{{L}_{\mathrm{no}}}
\newcommand{\Linv}{{L}_{\mathrm{inv}}}

\newcommand{\Sigmat}{\mathsf{\Sigma}_2^p}
\newcommand{\Pit}{\mathsf{\Pi}_2^p}
\newcommand{\QCSigma}{\mathsf{QC\Sigma}}
\newcommand{\QCPi}{\mathsf{QC\Pi}}
\newcommand{\QSigma}{\mathsf{Q\Sigma}}
\newcommand{\QPi}{\mathsf{Q\Pi}}
\newcommand{\QSigmapure}{\mathsf{pureQ\Sigma}}
\newcommand{\QSigmapureSEP}{\mathsf{pureQ\Sigma}^{\mathsf{SEP}}}
\newcommand{\QPipure}{\mathsf{pureQ\Pi}}
\newcommand{\QPiSEP}{\mathsf{Q\Pi}^{\mathsf{SEP}}}
\newcommand{\QPipureSEP}{\mathsf{pureQ\Pi}^{\mathsf{SEP}}}

\renewcommand{\P}{\mathsf{P}}
\def\complex{\mathbb{C}}

\def\natural{\mathbb{N}}

\def\({\left(}
\def\){\right)}
\def\X{\mathcal{X}}
\def\Y{\mathcal{Y}}
\def\Z{\mathcal{Z}}

\def\yes{\text{yes}}
\def\no{\text{no}}
\newcommand{\class}[1]{\textup{#1}}
\newcommand{\PrQCMA}{\class{Precise}\text{-}\QCMA}
\newcommand{\mpoly}{\textup{mpoly}}
\newcommand{\Bpoly}{\class{BQP$_{/\mpoly}$}}

\newcommand{\pL}{p_L}
\newcommand{\VLn}{V_{n}}
\newcommand{\VLx}{V_{|x|}}
\newcommand{\VLxuq}{V^{\mathrm{uq}}_{|x|}}
\newcommand{\brackets}[1]{\ensuremath{\left( #1 \right)}}
\newcommand{\Brackets}[1]{\ensuremath{\left[ #1 \right]}}

\newcommand{\bit}{\left\{ 0, 1 \right\}}
\newcommand{\qcma}{\mathsf{QCMA}}
\newcommand{\uqcma}{\mathsf{UQCMA}}
\newcommand{\UQCMA}{\mathsf{UQCMA}}
\newcommand{\tqcmapp}{\mathsf{TQCMAPP}}
\newcommand{\uqcmapp}{\mathsf{UQCMAPP}}
\newcommand{\TQCMAPP}{\tqcmapp}
\newcommand{\UQCMAPP}{\uqcmapp}
\newcommand{\tqcmappyes}{\TQCMAPP_{\mathrm{yes}}}
\newcommand{\uqcmappyes}{\UQCMAPP_{\mathrm{yes}}}
\newcommand{\tqcmappno}{\TQCMAPP_{\mathrm{no}}}
\newcommand{\uqcmappno}{\UQCMAPP_{\mathrm{no}}}
\newcommand{\etal}{{et al.}}
\newcommand{\trdist}[1]{\norm{#1}_{\mathrm{tr}}}
\newcommand{\approxbv}{\mathsf{ApproxBV}}
\newcommand{\QCAM}{\mathsf{QCAM}}

\newcommand{\sptwo}{\mathsf{S}^p_2}
\newcommand{\osptwo}{\mathsf{OS}^p_2}
\newcommand{\qcsptwo}{\mathsf{QCS}^p_2}

%% file: AGKR_QPH.bib
@article{GY23,
  author       = {Sabee Grewal and Justin Yirka},
  title        = {The Entangled Quantum Polynomial Hierarchy Collapses},
  journal      = {CoRR},
  year         = {2023}
}

@article{FGN23,
  author       = {Chirag Falor and
                  Shu Ge and
                  Anand Natarajan},
  title        = {A Collapsible Polynomial Hierarchy for Promise Problems},
  journal      = {CoRR},
  volume       = {abs/2311.12228},
  year         = {2023},
  url          = {https://doi.org/10.48550/arXiv.2311.12228},
  doi          = {10.48550/ARXIV.2311.12228},
  eprinttype    = {arXiv},
  eprint       = {2311.12228},
  bibsource    = {dblp computer science bibliography, https://dblp.org}
}

@article{KobWat98,
  author    = {Johannes K{\"{o}}bler and
               Osamu Watanabe},
  title     = {New Collapse Consequences of {NP} Having Small Circuits},
  journal   = {{SIAM} J. Comput.},
  volume    = {28},
  number    = {1},
  pages     = {311--324},
  year      = {1998},
  url       = {https://doi.org/10.1137/S0097539795296206},
  doi       = {10.1137/S0097539795296206},
  bibsource = {dblp computer science bibliography, https://dblp.org}
}

@inproceedings{CR06,
author = {Chakaravarthy, Venkatesan T. and Roy, Sambuddha},
title = {Oblivious Symmetric Alternation},
year = {2006},
isbn = {3540323015},
publisher = {Springer-Verlag},
address = {Berlin, Heidelberg},
url = {https://doi.org/10.1007/11672142_18},
doi = {10.1007/11672142_18},
pages = {230–241},
numpages = {12},
location = {Marseille, France},
series = {STACS'06}
}

@article{ABOBS22,
  doi = {10.22331/q-2022-03-17-668},
  url = {https://doi.org/10.22331/q-2022-03-17-668},
  title = {The {P}ursuit of {U}niqueness: {E}xtending {V}aliant-{V}azirani {T}heorem to the {P}robabilistic and {Q}uantum {S}ettings},
  author = {Aharonov, Dorit and Ben-Or, Michael and Brand{\~{a}}o, Fernando G.S.L. and Sattath, Or},
  journal = {{Quantum}},
  issn = {2521-327X},
  publisher = {{Verein zur F{\"{o}}rderung des Open Access Publizierens in den Quantenwissenschaften}},
  volume = {6},
  pages = {668},
  month = mar,
  year = {2022}
}

@article{qph,
  author    = {Sevag Gharibian and
               Miklos Santha and
               Jamie Sikora and
               Aarthi Sundaram and
               Justin Yirka},
  title     = {Quantum generalizations of the polynomial hierarchy with applications
               to {QMA(2)}},
  journal   = {Comput. Complex.},
  volume    = {31},
  number    = {2},
  pages     = {13},
  year      = {2022},
  url       = {https://doi.org/10.1007/s00037-022-00231-8},
  doi       = {10.1007/s00037-022-00231-8},
  bibsource = {dblp computer science bibliography, https://dblp.org}
}

@article{AaronsonD14,
  author       = {Scott Aaronson and
                  Andrew Drucker},
  title        = {A Full Characterization of Quantum Advice},
  journal      = {{SIAM} J. Comput.},
  volume       = {43},
  number       = {3},
  pages        = {1131--1183},
  year         = {2014},
  url          = {https://doi.org/10.1137/110856939},
  doi          = {10.1137/110856939},
  bibsource    = {dblp computer science bibliography, https://dblp.org}
}

@article{AaronsonCGK17,
  author       = {Scott Aaronson and
                  Alexandru Cojocaru and
                  Alexandru Gheorghiu and
                  Elham Kashefi},
  title        = {On the implausibility of classical client blind quantum computing},
  journal      = {CoRR},
  volume       = {abs/1704.08482},
  year         = {2017},
  url          = {http://arxiv.org/abs/1704.08482},
  eprinttype    = {arXiv},
  eprint       = {1704.08482},
  bibsource    = {dblp computer science bibliography, https://dblp.org}
}

@article{AKSS95,
title = {If NP has polynomial-size circuits, then MA = AM},
journal = {Theoretical Computer Science},
volume = {137},
number = {2},
pages = {279-282},
year = {1995},
issn = {0304-3975},
doi = {https://doi.org/10.1016/0304-3975(95)91133-B},
url = {https://www.sciencedirect.com/science/article/pii/030439759591133B},
author = {Vikraman Arvind and Johannes Köbler and Uwe Schöning and Rainer Schuler}
}

@inproceedings{aaronsonAcrobaticsBQP2022,
  title = {The {{Acrobatics}} of {{BQP}}},
  booktitle = {37th {{Computational Complexity Conference}} ({{CCC}} 2022)},
  author = {Aaronson, Scott and Ingram, DeVon and Kretschmer, William},
  editor = {Lovett, Shachar},
  year = {2022},
  series = {Leibniz {{International Proceedings}} in {{Informatics}} ({{LIPIcs}})},
  volume = {234},
  pages = {20:1--20:17},
  publisher = {{Schloss Dagstuhl \textendash{} Leibniz-Zentrum f\"ur Informatik}},
  address = {{Dagstuhl, Germany}},
  issn = {1868-8969},
  doi = {10.4230/LIPIcs.CCC.2022.20},
  urldate = {2023-09-05},
  isbn = {978-3-95977-241-9}
}

@inproceedings{aaronsonComputationalComplexityLinear2011,
  title = {The {{Computational Complexity}} of {{Linear Optics}}},
  booktitle = {Forty-Third {{Annual ACM Symposium}} on {{Theory}} of {{Computing}}},
  author = {Aaronson, Scott and Arkhipov, Alex},
  year = {2011},
  series = {{{STOC}} '11},
  pages = {333--342},
  publisher = {{ACM}},
  address = {{New York, NY, USA}},
  doi = {10.1145/1993636.1993682},
  isbn = {978-1-4503-0691-1}
}

@inproceedings{aharonovQuantumCircuitsMixed1998,
  title = {Quantum Circuits with Mixed States},
  booktitle = {Proceedings of the Thirtieth Annual {{ACM}} Symposium on {{Theory}} of Computing},
  author = {Aharonov, Dorit and Kitaev, Alexei and Nisan, Noam},
  year = {1998},
  month = may,
  series = {{{STOC}} '98},
  pages = {20--30},
  publisher = {{Association for Computing Machinery}},
  address = {{New York, NY, USA}},
  doi = {10.1145/276698.276708},
  urldate = {2023-05-09},
  isbn = {978-0-89791-962-3}
}

@inproceedings{bittelOptimalDepthVariational2023,
  title = {The {{Optimal Depth}} of {{Variational Quantum Algorithms Is QCMA-Hard}} to {{Approximate}}},
  booktitle = {38th {{Computational Complexity Conference}} ({{CCC}} 2023)},
  author = {Bittel, Lennart and Gharibian, Sevag and Kliesch, Martin},
  editor = {{Ta-Shma}, Amnon},
  year = {2023},
  series = {Leibniz {{International Proceedings}} in {{Informatics}} ({{LIPIcs}})},
  volume = {264},
  pages = {34:1--34:24},
  publisher = {{Schloss Dagstuhl \textendash{} Leibniz-Zentrum f\"ur Informatik}},
  address = {{Dagstuhl, Germany}},
  issn = {1868-8969},
  doi = {10.4230/LIPIcs.CCC.2023.34},
  urldate = {2023-09-21},
  isbn = {978-3-95977-282-2}
}

@article{boulandComplexityVerificationQuantum2019,
  title = {On the Complexity and Verification of Quantum Random Circuit Sampling},
  author = {Bouland, Adam and Fefferman, Bill and Nirkhe, Chinmay and Vazirani, Umesh},
  year = {2019},
  month = feb,
  journal = {Nature Physics},
  volume = {15},
  number = {2},
  pages = {159--163},
  publisher = {{Nature Publishing Group}},
  issn = {1745-2481},
  doi = {10.1038/s41567-018-0318-2},
  urldate = {2023-01-12},
  copyright = {2018 The Author(s), under exclusive licence to Springer Nature Limited},
  langid = {english}
}

@article{bremnerClassicalSimulationCommuting2010,
  title = {Classical Simulation of Commuting Quantum Computations Implies Collapse of the Polynomial Hierarchy},
  author = {Bremner, Michael J. and Jozsa, Richard and Shepherd, Dan J.},
  year = {2010},
  month = aug,
  journal = {Proceedings of the Royal Society A: Mathematical, Physical and Engineering Sciences},
  volume = {467},
  number = {2126},
  pages = {459--472},
  publisher = {{Royal Society}},
  doi = {10.1098/rspa.2010.0301},
  urldate = {2023-09-03}
}

@article{furstParityCircuitsPolynomialtime1984,
  title = {Parity, Circuits, and the Polynomial-Time Hierarchy},
  author = {Furst, Merrick and Saxe, James B. and Sipser, Michael},
  year = {1984},
  month = dec,
  journal = {Mathematical systems theory},
  volume = {17},
  number = {1},
  pages = {13--27},
  issn = {1433-0490},
  doi = {10.1007/BF01744431},
  urldate = {2023-09-03},
  langid = {english}
}

@inproceedings{gharibianHardnessApproximationQuantum2012,
  title = {Hardness of {{Approximation}} for {{Quantum Problems}}},
  booktitle = {Automata, {{Languages}}, and {{Programming}}},
  author = {Gharibian, Sevag and Kempe, Julia},
  editor = {Czumaj, Artur and Mehlhorn, Kurt and Pitts, Andrew and Wattenhofer, Roger},
  year = {2012},
  series = {Lecture {{Notes}} in {{Computer Science}}},
  pages = {387--398},
  publisher = {{Springer}},
  address = {{Berlin, Heidelberg}},
  doi = {10.1007/978-3-642-31594-7_33},
  isbn = {978-3-642-31594-7},
  langid = {english}
}

@article{gharibianQuantumGeneralizationsPolynomial2022,
  title = {Quantum Generalizations of the Polynomial Hierarchy with Applications to {{QMA}}(2)},
  author = {Gharibian, Sevag and Santha, Miklos and Sikora, Jamie and Sundaram, Aarthi and Yirka, Justin},
  year = {2022},
  month = sep,
  journal = {computational complexity},
  volume = {31},
  number = {2},
  pages = {13},
  issn = {1420-8954},
  doi = {10.1007/s00037-022-00231-8},
  urldate = {2023-09-03},
  langid = {english}
}

@incollection{goldreichPromiseProblemsSurvey2006,
  title = {On {{Promise Problems}}: {{A Survey}}},
  shorttitle = {On {{Promise Problems}}},
  booktitle = {Theoretical {{Computer Science}}: {{Essays}} in {{Memory}} of {{Shimon Even}}},
  author = {Goldreich, Oded},
  editor = {Goldreich, Oded and Rosenberg, Arnold L. and Selman, Alan L.},
  year = {2006},
  series = {Lecture {{Notes}} in {{Computer Science}}},
  pages = {254--290},
  publisher = {{Springer}},
  address = {{Berlin, Heidelberg}},
  doi = {10.1007/11685654_12},
  urldate = {2023-09-21},
  isbn = {978-3-540-32881-0},
  langid = {english}
}

@article{harrowTestingProductStates2013a,
  title = {Testing {{Product States}}, {{Quantum Merlin-Arthur Games}} and {{Tensor Optimization}}},
  author = {Harrow, Aram W. and Montanaro, Ashley},
  year = {2013},
  month = feb,
  journal = {Journal of the ACM},
  volume = {60},
  number = {1},
  pages = {3:1--3:43},
  issn = {0004-5411},
  doi = {10.1145/2432622.2432625},
  urldate = {2023-09-21}
}

@inproceedings{iraniQuantumSearchToDecisionReductions2022,
  title = {Quantum {{Search-To-Decision Reductions}} and the {{State Synthesis Problem}}},
  booktitle = {37th {{Computational Complexity Conference}} ({{CCC}} 2022)},
  author = {Irani, Sandy and Natarajan, Anand and Nirkhe, Chinmay and Rao, Sujit and Yuen, Henry},
  editor = {Lovett, Shachar},
  year = {2022},
  series = {Leibniz {{International Proceedings}} in {{Informatics}} ({{LIPIcs}})},
  volume = {234},
  pages = {5:1--5:19},
  publisher = {{Schloss Dagstuhl \textendash{} Leibniz-Zentrum f\"ur Informatik}},
  address = {{Dagstuhl, Germany}},
  issn = {1868-8969},
  doi = {10.4230/LIPIcs.CCC.2022.5},
  urldate = {2023-09-21},
  isbn = {978-3-95977-241-9}
}

@article{j.lockhartQuantumStateIsomorphism2017,
  title = {Quantum {{State Isomorphism}}},
  author = {{J. Lockhart} and {C. E. Gonz\'alez-Guill\'en}},
  year = {2017},
  journal = {arXiv preprint arXiv:1709.09622},
  eprint = {1709.09622},
  archiveprefix = {arxiv}
}

@inproceedings{jainParallelApproximationNoninteractive2009,
  title = {Parallel {{Approximation}} of {{Non-interactive Zero-sum Quantum Games}}},
  booktitle = {Proceedings of the 24th {{Annual IEEE Conference}} on {{Computational Complexity}}, {{CCC}} 2009, {{Paris}}, {{France}}, 15-18 {{July}} 2009},
  author = {Jain, Rahul and Watrous, John},
  year = {2009},
  pages = {243--253},
  publisher = {{IEEE Computer Society}},
  doi = {10.1109/CCC.2009.26}
}

@inproceedings{karpConnectionsNonuniformUniform1980,
  title = {Some {{Connections Between Nonuniform}} and {{Uniform Complexity Classes}}},
  booktitle = {Twelfth {{Annual ACM Symposium}} on {{Theory}} of {{Computing}}},
  author = {Karp, Richard M. and Lipton, Richard J.},
  year = {1980},
  series = {{{STOC}} '80},
  pages = {302--309},
  publisher = {{ACM}},
  address = {{New York, NY, USA}},
  doi = {10.1145/800141.804678},
  isbn = {0-89791-017-6}
}

@misc{kobayashiQuantumCertificateVerification2001,
  title = {Quantum {{Certificate Verification}}: {{Single}} versus {{Multiple Quantum Certificates}}},
  shorttitle = {Quantum {{Certificate Verification}}},
  author = {Kobayashi, Hirotada and Matsumoto, Keiji and Yamakami, Tomoyuki},
  year = {2001},
  month = oct,
  number = {arXiv:quant-ph/0110006},
  eprint = {quant-ph/0110006},
  publisher = {{arXiv}},
  doi = {10.48550/arXiv.quant-ph/0110006},
  urldate = {2023-09-04},
  archiveprefix = {arxiv}
}

@article{lautemannBPPPolynomialHierarchy1983,
  title = {{{BPP}} and the Polynomial Hierarchy},
  author = {Lautemann, Clemens},
  year = {1983},
  month = nov,
  journal = {Information Processing Letters},
  volume = {17},
  number = {4},
  pages = {215--217},
  issn = {0020-0190},
  doi = {10.1016/0020-0190(83)90044-3},
  urldate = {2023-09-21}
}

@article{mintertConcurrenceMixedMultipartite2005,
  title = {Concurrence of {{Mixed Multipartite Quantum States}}},
  author = {Mintert, Florian and Ku{\'s}, Marek and Buchleitner, Andreas},
  year = {2005},
  month = dec,
  journal = {Physical Review Letters},
  volume = {95},
  number = {26},
  pages = {260502},
  publisher = {{American Physical Society}},
  doi = {10.1103/PhysRevLett.95.260502},
  urldate = {2023-08-11}
}

@inproceedings{sipserComplexityTheoreticApproach1983,
  title = {A Complexity Theoretic Approach to Randomness},
  booktitle = {Proceedings of the Fifteenth Annual {{ACM}} Symposium on {{Theory}} of Computing},
  author = {Sipser, Michael},
  year = {1983},
  month = dec,
  series = {{{STOC}} '83},
  pages = {330--335},
  publisher = {{Association for Computing Machinery}},
  address = {{New York, NY, USA}},
  doi = {10.1145/800061.808762},
  urldate = {2023-09-21},
  isbn = {978-0-89791-099-6}
}

@misc{soleimanifarTestingMatrixProduct2022,
  title = {Testing Matrix Product States},
  author = {Soleimanifar, Mehdi and Wright, John},
  year = {2022},
  month = jan,
  number = {arXiv:2201.01824},
  eprint = {2201.01824},
  primaryclass = {quant-ph},
  publisher = {{arXiv}},
  doi = {10.48550/arXiv.2201.01824},
  urldate = {2023-06-20},
  archiveprefix = {arxiv}
}

@article{stockmeyerPolynomialtimeHierarchy1976,
  title = {The Polynomial-Time Hierarchy},
  author = {Stockmeyer, Larry J.},
  year = {1976},
  journal = {Theoretical Computer Science},
  volume = {3},
  number = {1},
  pages = {1--22},
  issn = {0304-3975},
  doi = {http://dx.doi.org/10.1016/0304-3975(76)90061-X}
}

@article{todaPPHardPolynomialTime1991,
  title = {{{PP}} Is as {{Hard}} as the {{Polynomial-Time Hierarchy}}},
  author = {Toda, Seinosuke},
  year = {1991},
  month = oct,
  journal = {SIAM Journal on Computing},
  volume = {20},
  number = {5},
  pages = {865--877},
  publisher = {{Society for Industrial and Applied Mathematics}},
  issn = {0097-5397},
  doi = {10.1137/0220053},
  urldate = {2023-08-22}
}

@article{valiantNPEasyDetecting1986,
  title = {{{NP}} Is as Easy as Detecting Unique Solutions},
  author = {Valiant, L. G. and Vazirani, V. V.},
  year = {1986},
  month = jan,
  journal = {Theoretical Computer Science},
  volume = {47},
  pages = {85--93},
  issn = {0304-3975},
  doi = {10.1016/0304-3975(86)90135-0},
  urldate = {2023-09-21}
}

@incollection{yamakamiQuantumNPQuantum2002,
  title = {Quantum {{NP}} and a {{Quantum Hierarchy}}},
  booktitle = {Foundations of {{Information Technology}} in the {{Era}} of {{Network}} and {{Mobile Computing}}: {{IFIP}} 17th {{World Computer Congress}} \textemdash{} {{TC1 Stream}} / 2nd {{IFIP International Conference}} on {{Theoretical Computer Science}} ({{TCS}} 2002) {{August}} 25\textendash 30, 2002, {{Montr\'eal}}, {{Qu\'ebec}}, {{Canada}}},
  author = {Yamakami, Tomoyuki},
  editor = {{Baeza-Yates}, Ricardo and Montanari, Ugo and Santoro, Nicola},
  year = {2002},
  series = {{{IFIP}} \textemdash{} {{The International Federation}} for {{Information Processing}}},
  pages = {323--336},
  publisher = {{Springer US}},
  address = {{Boston, MA}},
  doi = {10.1007/978-0-387-35608-2_27},
  urldate = {2023-09-21},
  isbn = {978-0-387-35608-2},
  langid = {english}
}

@inproceedings{C01,
  author       = {Jin{-}yi Cai},
  title        = {$S^{\mbox{p}}_{\mbox{2}} \subseteq ZPP^{\mbox{NP}}$},
  booktitle    = {42nd Annual Symposium on Foundations of Computer Science, {FOCS} 2001,
                  14-17 October 2001, Las Vegas, Nevada, {USA}},
  pages        = {620--629},
  publisher    = {{IEEE} Computer Society},
  year         = {2001},
  url          = {https://doi.org/10.1109/SFCS.2001.959938},
  doi          = {10.1109/SFCS.2001.959938},
  bibsource    = {dblp computer science bibliography, https://dblp.org}
}

@article{RS98,
  author       = {Alexander Russell and
                  Ravi Sundaram},
  title        = {Symmetric Alternation Captures {BPP}},
  journal      = {Comput. Complex.},
  volume       = {7},
  number       = {2},
  pages        = {152--162},
  year         = {1998},
  url          = {https://doi.org/10.1007/s000370050007},
  doi          = {10.1007/S000370050007},
  bibsource    = {dblp computer science bibliography, https://dblp.org}
}

@article{Can96,
  author       = {Ran Canetti},
  title        = {More on {BPP} and the Polynomial-Time Hierarchy},
  journal      = {Inf. Process. Lett.},
  volume       = {57},
  number       = {5},
  pages        = {237--241},
  year         = {1996},
  url          = {https://doi.org/10.1016/0020-0190(96)00016-6},
  doi          = {10.1016/0020-0190(96)00016-6},
  bibsource    = {dblp computer science bibliography, https://dblp.org}
}

@article{BCGKT96,
  author       = {Nader H. Bshouty and
                  Richard Cleve and
                  Ricard Gavald{\`{a}} and
                  Sampath Kannan and
                  Christino Tamon},
  title        = {Oracles and Queries That Are Sufficient for Exact Learning},
  journal      = {Electron. Colloquium Comput. Complex.},
  volume       = {{TR95-015}},
  year         = {1995},
  url          = {https://eccc.weizmann.ac.il/eccc-reports/1995/TR95-015/index.html},
  eprinttype    = {ECCC},
  eprint       = {TR95-015},
  bibsource    = {dblp computer science bibliography, https://dblp.org}
}
